\documentclass[runningheads,envcountsect,envcountsame]{llncs}
\usepackage[utf8]{inputenc}
\RequirePackage[T1]{fontenc}
\usepackage{amssymb}
\usepackage{etoolbox}
\AtEndEnvironment{proof}{\phantom{}\qed}
\usepackage{ifmtarg}
\usepackage{tikz}
\usepackage{xspace, soul, enumerate}
\usepackage{lineno}

\usepackage{algorithm}
\usepackage[noend]{algpseudocode}
\usepackage{mathtools}
\usepackage{hyperref}

\newif\ifcomments
\newif\ifchanges
\commentsfalse\changesfalse

\newcommand{\Member}[1]{\myproblem{Member}(#1)}
\newcommand{\RangeMember}[1]{\myproblem{RangeMember}(#1)}

\newcommand{\RangeEval}[1]{\myproblem{RangeEval}(#1)}
\newcommand{\StringEquality}{\myproblem{StringEquality}}
\newcommand{\qrange}{\mtext{range}}
\newcommand{\qmember}{\mtext{member}}
\newcommand{\qequals}{\mtext{equals}}

\newcommand{\tsucc}[2]{\ensuremath{\lceil #1 \rceil^{#2}}}
\newcommand{\tpred}[2]{\ensuremath{\lfloor #1 \rfloor_{#2}}}

\newenvironment{qclist}{\begin{tabular}[t]{@{}rll@{}}}{\end{tabular}}
\newcommand{\qclitem}[3]{#1 & #2 & #3 \\}

\newcommand{\tpl}{\bar}

\newcommand{\mtext}[1]{\textsc{#1}}
\newcommand{\ins}{\mtext{ins}\xspace}
\newcommand{\del}{\mtext{del}\xspace} \newcommand{\set}{\mtext{set}\xspace}
\newcommand{\reset}{\mtext{reset}\xspace}

\newcommand{\schema}{\ensuremath{\tau}\xspace}

\newcommand{\domain}{\ensuremath{D}}
\newcommand{\query}{\ensuremath{Q}}

\newcommand{\db}{\ensuremath{\calD}\xspace}
\newcommand{\inp}{\ensuremath{\calI}\xspace}
\newcommand{\aux}{\ensuremath{\calA}\xspace}

\newcommand{\word}[1]{\text{word}(#1)}

\newcommand{\N}{\ensuremath{\mathbb{N}}}

\newcommand{\bigO}{\ensuremath{\mathcal{O}}}

\newcommand{\df}{\ensuremath{\mathrel{\smash{\stackrel{\scriptscriptstyle{
    \text{def}}}{=}}}} \;}
    
\makeatletter 
\newcommand{\ut}[4]{
  \@ifmtarg{#4}{t^{#1}_{#2}(#3) }{t^{#1}_{#2}(#3; #4)}
}

\newcommand{\ite}[3]{
  \@ifmtarg{#1}{
    \mtext{ITE}
   }{
    \mtext{ITE}\text{$(#1,#2,#3)$}  
  }
}
\makeatother

\DeclarePairedDelimiter\size{\lvert}{\rvert}

\newcommand{\ifte}[3]{\textbf{if }#1\textbf{ then }#2\textbf{ else }#3}

\newcommand{\fordo}[2]{\textbf{for }#1\textbf{ pardo }#2}
\newcommand{\ret}[1]{\textbf{return }#1}

\newcommand{\true}{\texttt{true}}
\newcommand{\false}{\texttt{false}}
\newcommand{\algexists}[2]{\textbf{exists}(#1 \ensuremath{\mid} #2)}
\newcommand{\algunique}[2]{\textbf{getUnique}(#1 \ensuremath{\mid} #2)}
\newcommand{\algfor}[2]{\textbf{for }#1\textbf{ do }#2}
\newcommand{\cond}{\texttt{condition}}
\newcommand{\term}{\texttt{term}}
\newcommand{\instance}{\texttt{instance}}

\newcommand  {\myclass} [1]  {\ensuremath{\textsf{\upshape #1}}}

\newcommand{\StaClass}[1]{\myclass{#1}\xspace}

\newcommand{\DynClass}[1]{\myclass{Dyn#1}\xspace}

\newcommand  {\myproblem} [1] {\normalfont{\textsc{#1}}\xspace}

\newcommand{\dynProblemDefinition}[4]{
    \setlength{\tabcolsep}{1.5pt}
    \begin{flushleft}
        \begin{tabular}[t]{r p{0.8\textwidth}}
            \textbf{Problem:} & #1 \\
            \textbf{Input:} & #2 \\
            \textbf{Changes:} & #3 \\
            \textbf{Queries:} & #4 \\
        \end{tabular}
    \end{flushleft}
}

\newcommand     {\AC}   {\StaClass{AC}}

\newcommand{\FO}{\StaClass{FO}}

\newcommand{\FOar}{\StaClass{FO$(\leq,+,\times)$}}

\newcommand{\CQ}[1][]{\StaClass{CQ}}
\newcommand{\UCQ}[1][]{\StaClass{UCQ}}
\newcommand{\CQneg}[1][]{\StaClass{CQ\ensuremath{^{\mneg}}}}
\newcommand{\UCQneg}[1][]{\StaClass{UCQ\ensuremath{^{\mneg}}}}

\newcommand{\mneg}{\neg} 

\newcommand{\DynProp}{\DynClass{Prop}}

\newcommand{\DynFO}{\DynClass{FO}}

\newcommand{\DynFOar}{\DynClass{FO$(\leq,+,\times)$}}

\spnewtheorem{goal}[theorem]{Goal}{\bfseries}{\itshape}

\let\llncsproof\proof
\renewcommand{\proof}[1][]{%
  \ifx!#1!\else\renewcommand{\proofname}{#1}\fi
  \llncsproof
}

\newenvironment{proofsketch}{\begin{proof}[Proof sketch]}{\end{proof}}

\providecommand {\calA}      {{\mathcal A}\xspace}

\providecommand {\calC}      {{\mathcal C}\xspace}
\providecommand {\calD}      {{\mathcal D}\xspace}

\providecommand {\calI}      {{\mathcal I}\xspace}

\providecommand {\calP}      {{\mathcal P}\xspace}

\providecommand {\calS}      {{\mathcal S}\xspace}

 \newcommand{\changeRuleWhere}[5]{\textbf{on change}\ #1\ \textbf{update}\ #2\ \textbf{at}\ #3\ \textbf{as}\ #4\ \textbf{where}\ #5} 

 \newcommand{\queryRule}[2]{\textbf{on query}\ #1\ \textbf{yield}\ #2}
 
\newcommand{\prog}{\ensuremath{\calP}\xspace}
\newcommand{\prob}{\ensuremath{\Pi}\xspace}

\algnewcommand\algorithmiconchange{\textbf{on change}}
\algnewcommand\algorithmiconquery{\textbf{on query}}
\algnewcommand\algorithmicupdate{\textbf{update}}
\algnewcommand\algorithmicat{\textbf{at}}
\algnewcommand\algorithmicby{\textbf{by}}
\algnewcommand\algorithmicpardo{\textbf{pardo}}
\algnewcommand\algorithmicwhere{\textbf{where}}
\algnewcommand\algorithmicwith{\textbf{with}}

\algnewcommand\algorithmicunique{\textbf{unique}}
\algnewcommand\algorithmicmin{\textbf{min}}
\algnewcommand\algorithmicmax{\textbf{max}}

\algnewcommand{\False}{\textbf{false}}
\algnewcommand{\True}{\textbf{true}}
\algnewcommand{\To}{\textbf{to}}
\algnewcommand{\Select}[4]{\State \ensuremath{#1 \gets #2(#3 \mid #4)}}
\algnewcommand{\Unique}[3]{\Select{#1}{\algorithmicunique}{#2}{#3}}
\algnewcommand{\Min}[3]{\Select{#1}{\algorithmicmin}{#2}{#3}}
\algnewcommand{\Max}[3]{\Select{#1}{\algorithmicmax}{#2}{#3}}

\algblockdefx{OnChangeAtWhereBy}{EndOnChangeAtWhereBy}%
    [4]{\algorithmiconchange\ #1\ \algorithmicupdate\ #2 \algorithmicat\ #3, \algorithmicforall\ #4, \algorithmicby:}%
    {\algorithmicend\ \algorithmiconchange}
    
\algblockdefx{UpdateAtWhereBy}{EndUpdateAtWhereBy}%
    [3]{\algorithmicupdate\ #1 \algorithmicat\ #2, \algorithmicforall\ #3, \algorithmicby:}%
    {\algorithmicend\ \algorithmiconchange}    

    \algblockdefx{UpdateAtBy}{EndUpdateAtWhereBy}%
    [2]{\algorithmicupdate\ #1 \algorithmicat\ #2, \algorithmicby:}%
    {\algorithmicend\ \algorithmiconchange}    
    
\algblockdefx{OnQuery}{EndOnQuery}%
    [1]{\algorithmiconquery\ #1:}%
    {\algorithmicend\ \algorithmiconchange}    

\algblockdefx{OnChange}{EndOnChange}%
    [1]{\algorithmiconchange\ #1}%
    {\algorithmicend\ \algorithmiconchange}

    \algblockdefx{OnChangeWith}{EndOnChangeWith}%
    [2]{\algorithmiconchange\ #1 \algorithmicwith\ #2}%
    {\algorithmicend\ \algorithmiconchange}
 
\algblockdefx{With}{EndWith}%
    {\algorithmicwith}%
    {\algorithmicend\ \algorithmiconchange}    

\algblockdefx{AtWhereBy}{EndAtWhereBy}%
    [2]{\algorithmicat\ #1\ \algorithmicwhere\ #2\ \algorithmicby:}%
    {\algorithmicend\ \algorithmiconchange}

\algblockdefx{By}{EndBy}%
    [0]{\algorithmicby:}%
    {\algorithmicend\ \algorithmiconchange}

\algblockdefx{Parfor}{EndFor}%
    [1]{\algorithmicfor\ #1\ \algorithmicpardo}%
    {\algorithmicend\ \algorithmicfor}

\algblockdefx{Parforall}{EndFor}%
    [1]{\algorithmicforall\ #1\ \algorithmicpardo}%
    {\algorithmicend\ \algorithmicfor}

\makeatletter
\ifthenelse{\equal{\ALG@noend}{t}}%
    {
    		\algtext*{EndOnChangeAtWhereBy}
        \algtext*{EndOnChange}
        \algtext*{EndOnQuery}
        \algtext*{EndAtWhereBy}
        \algtext*{EndBy}
        \algtext*{EndFor}
        \algtext*{EndWith}
        \algtext*{EndOnChangeWith}
        \algtext*{EndUpdateAtWhereBy}
    }
    {}%
\makeatother

\newcommand{\auxSchema}{\ensuremath{\schema_{\text{aux}}}\xspace}

\newcommand{\workbound}{w}

\ifcomments
\newcommand{\commentbox}[1]{\noindent\framebox{\parbox{0.98\linewidth}{#1}}}

\newcommand{\acomment}[2]{\ \\ \fbox{\parbox{0.98\linewidth}{{\sc #1}: #2}}}
\newcommand{\mcomment}[2]{{\color{blue}(#1)}\footnote{#1: #2}} 
\else
\newcommand{\commentbox}[1]{}
\newcommand{\mcomment}[2]{}
\newcommand{\acomment}[2]{}
\fi

\ifchanges

\setul{}{0.2mm}
\setstcolor{red}

\else

\fi

 \newcommand{\tzm}[1]{\mcomment{TZ}{#1}}
 \newcommand{\tsm}[1]{\mcomment{TS}{#1}}
 \newcommand{\nilsm}[1]{\mcomment{NV}{#1}}

\newcommand{\longversion}[1]{#1}

\tikzstyle{mnode}=[
  circle,
  fill=\mnodefillcolor, 
  draw=\mnodedrawcolor,
  minimum size=6pt, 
  inner sep=0pt
]

\tikzstyle{mnodeinvisible}=[
  minimum size=6pt, 
  inner sep=0pt
]

\tikzstyle{invisible}=[
  minimum size=0pt, 
  inner sep=0pt
]

\tikzstyle{invisiblel}=[
  minimum size=10pt, 
  inner sep=0pt
]

\tikzstyle{invisibleEdge}=[
  transparent
]

\tikzstyle{nameNode}=[
  font=\scriptsize
]

\tikzstyle{namingNode}=[
  font=\normalsize
]

\tikzstyle{mEdge}=[
  -latex', 
  thick, 
  shorten >=3pt, 
  shorten <=3pt,
  draw=black!80,
]

\tikzstyle{dDashedEdge}=[
  -latex', 
  thick, 
  shorten >=3pt, 
  shorten <=3pt,
  draw=black!80,
  dashed
]

\tikzstyle{dEdge}=[
  -latex', 
  thick, 
  shorten >=1pt, 
  shorten <=1pt,
  draw=black!80,
]

\tikzstyle{dhEdge}=[
  -latex', 
  thick, 
  shorten >=3pt, 
  shorten <=3pt,
  draw=black!80,
]
\tikzstyle{uEdge}=[
  thick, 
  shorten >=3pt, 
  shorten <=3pt,
  draw=black!80,
]
\tikzstyle{uhEdge}=[
  thick, 
  shorten >=3pt, 
  shorten <=3pt,
  draw=black!80,
]

\tikzstyle{cEdge}=[
  ultra thick, 
  shorten >=3pt, 
  shorten <=3pt,
  draw=black!80,
]

\tikzstyle{dotsEdge}=[
  very thick, 
  loosely dotted, 
  shorten >=7pt, 
  shorten <=7pt
]

\tikzstyle{class rectangle}=[
  draw=black,
  inner sep=0.2cm,
  rounded corners=5pt,
  thick
]

\tikzstyle{mline}=[
  draw=black,
  inner sep=0.2cm,
  rounded corners=5pt,
  thick
]

\tikzstyle{mainclass rectangle}=[
  draw=blue,
  inner sep=0.2cm,
  rounded corners=5pt,
  very thick
]

\newcommand{\mnodedrawcolor}{black!80}
\newcommand{\mnodefillcolor}{black!40}

\pgfdeclarelayer{background}
\pgfdeclarelayer{substructure}
\pgfdeclarelayer{edges}
\pgfdeclarelayer{foreground}
\pgfsetlayers{background,substructure,edges,main,foreground}

\tikzstyle{background rectangle}=[
  fill=black!10,
  draw=black!20,
  line width = 5pt,
  inner sep=0.4cm,
  outer sep=0.4cm,
  rounded corners=5pt
]

\begin{document}
\title{Work-sensitive Dynamic Complexity of Formal Languages}
%
%


\author{Jonas Schmidt\inst{1}\and
Thomas Schwentick\inst{1}\and
Till Tantau\inst{2}\and
Nils Vortmeier\inst{3}\and
Thomas Zeume\inst{4}
}
\authorrunning{J. Schmidt, T. Schwentick, T. Tantau, N. Vortmeier,  and T. Zeume}
%
\institute{TU Dortmund University, Germany\\\email{\{jonas2.schmidt,thomas.schwentick\}@tu-dortmund.de} \and
Universität zu Lübeck, Germany\\\email{tantau@tcs.uni-luebeck.de}\and
University of Zurich, Switzerland\\\email{nils.vortmeier@uzh.ch} \and
Ruhr University Bochum, Germany\\\email{thomas.zeume@rub.de} }

\maketitle              
\begin{abstract}
Which amount of parallel resources is needed for updating a query result after changing an input? In this work we study the amount of work required for dynamically answering membership and range queries for formal languages in parallel constant time with polynomially many processors. As a prerequisite, we propose a framework for specifying dynamic, parallel, constant-time programs that require small amounts of work. This framework is based on the dynamic descriptive complexity framework by Patnaik and Immerman.

\end{abstract}
    
    \section{Introduction}\label{section:introduction}
    
Which amount of parallel resources is needed for updating a query result after changing an input, in particular if we only want to spend constant parallel time? 

In classical, non-dynamic computations, parallel constant time is well understood. Constant time on CRAMs, a variant of CRCW-PRAMs used by Immerman \cite{Immerman12}, corresponds to constant-depth in circuits, so, to the circuit class $\AC^0$, as well as to expressibility in first-order logic with built-in arithmetic (see, for instance, the books of Immerman \cite[Theorem 5.2]{Immerman12} and Vollmer \cite[Theorems 4.69 and 4.73]{Vollmer13}). Even more, the amount of work, that is, the overall number of operations of all processors, is connected  to the number of variables required by a first-order formula \cite[Theorem 5.10]{Immerman12}.

However, the work aspect of constant parallel time algorithms is less understood for scenarios where the input is subject to changes. To the best of our knowledge, there is only little previous work on constant-time PRAMs in dynamic scenarios. A notable exception is early work showing that spanning trees and connected components can be computed in constant time by CRCW-PRAMs with $O(n^4)$ and $O(n^2)$ processors, respectively~\cite{SherlekarPR85}.

In an orthogonal line of research, parallel dynamic constant time has been studied from a logical perspective in the dynamic complexity framework by Patnaik and Immerman \cite{PatnaikI94} and Dong, Su, and Topor \cite{DongT92,DongS93}. In this framework, the update of query results after a change is expressed by first-order formulas. The formulas may refer to auxiliary relations, whose updates in turn are also specified by first-order formulas (see Section \ref{section:model} for more details). The queries maintainable in this fashion constitute the dynamic complexity class $\DynFO$. 
Such queries can be updated by PRAMs in constant time with a polynomial number of processors. In this line of work, the main focus in recent years has been on proving that queries are in $\DynFO$, and thus emphasised the constant time aspect.  It has, for instance, been shown that all context-free languages \cite{GeladeMS12} and the reachability query \cite{DattaKMSZ18} are in $\DynFO$.

However, if one tries to make the ``\DynFO approach'' for dynamic problems relevant for practical considerations, the work that is needed to carry out the specified updates, hence the \emph{work} of a parallel algorithm implementing them, is a crucial factor. The current general polynomial upper bounds are too coarse. In this paper, we therefore  initiate the investigation of more work-efficient dynamic programs that can be specified by first-order logic and that can therefore be carried out by PRAMs in constant time. 
To do so, we propose a framework for specifying such dynamic, parallel, constant-time programs, which is based on the \DynFO framework, but allows for more precise (and better) bounds on the necessary work of a program.

\begin{goal}
  Extend the formal framework of dynamic complexity towards the consideration of parallel work.
\end{goal}
Towards this goal, we link the framework we propose to the CRAM framework  in Section \ref{section:model}. In fact, the new framework also takes a somewhat wider perspective, since it does not focus exclusively at one query under a set of change operations, but rather considers dynamic problems that may have several change and query operations (and could even have operations that combine the two). Therefore, from now on we speak about dynamic problems and not about (single) queries.

\begin{goal}
    Find work-efficient \DynFO-programs for dynamic problems that are known to be in \DynFO
    (but whose dynamic programs\footnote{In the field of dynamic complexity the term ``dynamic program'' is traditionally used for the programs for updating the auxiliary data after a change.
    The term should not be confused with the ``dynamic programming'' technique used in algorithm design.}
    are not competitive, work-wise).
\end{goal}
	
Ideally we aim at showing that dynamic problems can be maintained in $\DynFO$ with sublinear or even polylogarithmic work.
One line of attack for this goal is to study dynamic algorithms and to see whether they can be transformed into parallel $\bigO(1)$-time algorithms with small work. There is a plethora of work that achieves polylogarithmic sequential update time (even though, sometimes only amortised), see for instance \cite{AlstrupHusfeldt+2004,FrandsenMiltersen+1997,HolmLT01,HolmR20}. 
For many of these problems, it is known that they can be maintained in constant parallel time with polynomial work, e.g. as mentioned above, it has been shown that connectivity and maintenance of regular (and even context-free) languages is in $\DynFO$.


In  this paper, we follow this approach for dynamic string problems, more specifically, dynamic problems that allow membership and range queries for regular and context-free languages. Our results can be summarised as follows.

We show in Section~\ref{section:regular} that regular languages can be maintained in constant time with $\bigO(n^\epsilon)$ work for all $\epsilon > 0$ and that for star-free languages even work $\bigO(\log n)$ can be achieved. These results hold for range and membership queries.  \tsm{The discussion why star-free makes sense is postponed to the text. Less obvious now, since we say star-free...}

For context-free languages, the situation is not as nice, as we observe in Section~\ref{section:cfl}.
We show that subject to a well-known conjecture, we cannot hope for maintaining membership in general context-free languages in $\DynFO$ with less than $\bigO(n^{1.37-\epsilon})$ work.
The same statement holds even for the bound $\bigO(n^{2-\epsilon})$ and ``combinatorial dynamic programs''.
For Dyck languages,  that is, sets of well-formed strings of parentheses, we show that this barrier does not apply.
Their membership problem can be maintained with $\bigO(n (\log n)^3)$ work in general, and with polylogarithmic work if there is only one kind of parentheses.
By a different approach, range queries can be maintained with work $\bigO(n^{1+\epsilon})$ in general, and $\bigO(n^\epsilon)$ for one parenthesis type.

\emph{Related work.} A complexity theory of incremental time has been developed in~\cite{MiltersenSVT94}.
We discuss previous work on dynamic complexity of formal languages in Sections~\ref{section:regular} and \ref{section:cfl}. 


    \section{Preliminaries}\label{section:preliminaries}
\newcommand{\qd}{\text{qd}}

Since dynamic programs are based on first-order logic, we represent inputs like graphs and strings as well as ``internal'' data structures as logical structures. 

A\emph{
schema} $\schema$ consists of a set of relation symbols and function symbols with a corresponding arity. A constant symbol is a function symbol with arity $0$.
A \emph{structure} $\db$ over schema $\schema$ with finite domain $\domain$ has, for every $k$-ary relation symbol $R \in \schema$,
a relation $R^\db \subseteq \domain^k$, as well as a function $f^\db \colon D^k \to D$ for every $k$-ary function symbol $f \in \schema$.
We allow partially defined functions and write $f^\db(\tpl a) = \bot$ if $f^\db$ is not defined for $\tpl a$ in $\db$. Formally, this can be realized using an additional relation that contains the domain of $f^\db$. 
We occasionally also use functions $f^\db \colon D^k \to D^\ell$ for some $\ell > 1$. Formally, such a function represents $\ell$ functions $f^\db_1, \ldots, f^\db_\ell \colon D^k \to D$ with $f^\db(\tpl a) \df (f^\db_1(\tpl a),\ldots,f^\db_\ell(\tpl a))$.

Throughout this work, the structures we consider provide a linear order $\leq$ on their domain~$\domain$. As we can thus identify $\domain$ with an initial sequence of the natural numbers, we usually just assume that $\domain = [n] \df \{0, \ldots, n-1\}$ for some natural number $n$.

We assume familiarity with first-order logic $\FO$, and refer to \cite{Libkin04} for basics of Finite Model Theory.
In this paper, unless stated otherwise, first-order formulas \emph{always} have access to a linear order on the domain, as well as compatible functions $+$ and $\times$ that express addition and multiplication, respectively. This holds in particular for formulas in dynamic programs.
We use the following ``if-then-else'' construct
:
if $\varphi$ is a formula, and $t_1$ and $t_2$ are terms, then $\ite{\varphi}{t_1}{t_2}$ is a term.
Such a term evaluates to the result of $t_1$ if $\varphi$ is satisfied, otherwise to $t_2$. 
  


Following \cite{GeladeMS12}, we encode words of length (at most) $n$ over an alphabet $\Sigma$ by \emph{word structures}, that is, as relational structures $W$ with universe $\{0, \ldots, n-1\}$,
one unary relation $R_{\sigma}$ for each symbol $\sigma \in \Sigma$ and the canonical linear order $\leq$ on $\{0, \ldots, n-1\}$.
We only consider structures for which, for every position $i$, $R_{\sigma}(i)$ holds for at most one $\sigma \in \Sigma$ and
write $W(i) = \sigma$ if $R_{\sigma}(i)$ holds and $W(i) = \epsilon$ if no such $\sigma$ exists.
We write $\word{W}$ for the word represented by $W$, that is, the concatenation $w = W(0) \circ \ldots \circ W(n-1)$. As an example, the word structure $W_0$ with domain $\{0,1,2,3\}$, $W(1)=a$, $W(3)=b$ and $W(0)=W(2)=\epsilon$ represents the string $ab$. We write $\word{W}[\ell,r]$ for the word  $W(\ell) \circ \ldots \circ W(r)$.

Informally, a \emph{dynamic problem} can be seen as a data type: it consists of some underlying structure together with a set $\Delta$ of operations. We distinguish between \emph{change operations} that can modify the structure and \emph{query operations} that yield information about the structure, but combined operations could be allowed, as well. 
Thus, a dynamic problem is characterised by the schema of its underlying structures and the operations that it supports.\footnote{This view is a bit broader than the traditional setting of Dynamic Complexity, where there can be various change operations but usually only one fixed query is supported.}

In this paper, we are particularly interested in dynamic language problems, defined as follows. Words are represented as word structures $W$ with elementary change operations $\set_\sigma(i)$ (with the effect that $W(i)$ becomes $\sigma$ if it was $\epsilon$ before) and  $\reset(i)$ (with the effect that $W(i)$ becomes $\epsilon$).

For some fixed language $L$ over some alphabet $\Sigma$, the dynamic problem $\RangeMember{L}$ further supports one query operation $\qrange(\ell,r)$. It yields the result true, if $\word{W}[\ell,r]$   is in $L$, and otherwise false.\tsm{Please let us use $\ell$ instead of $l$.}

In the following, we denote a word structure $W$ as a sequence $w_0 \ldots w_{n-1}$ of letters with $w_i \in \Sigma \cup \{\epsilon\}$ in order to have an easier, less formal notation.
Altogether, the dynamic problem $\RangeMember{L}$ is defined as follows.\tsm{I wonder whether it should be renamed, e.g.,  as DynLang or DynRangeLang. Or DynRangeMember?}\tsm{I hope we can have some background colours here, like in tables. }\tzm{need some formatting}

\dynProblemDefinition{$\RangeMember{L}$}
    {A sequence $w = w_0 \ldots w_{n-1}$ of letters with $w_i \in \Sigma \cup \{\epsilon\}$}
    {\begin{qclist}
        \qclitem{$\set_\sigma(i)$ for $\sigma \in \Sigma$:}{}{Sets $w_i$ to $\sigma$, if $w_i = \epsilon$}
        \qclitem{$\reset(i)$:}{}{Sets $w_i$ to $\epsilon$}
    \end{qclist}}
    {\begin{qclist}
        \qclitem{$\qrange(\ell,r)$:}{}{Is $w_\ell \circ \cdots \circ w_r \in L$?}
    \end{qclist}}
In this example, the query $\qrange$ maps  (binary) pairs of domain elements to a truth value and thus  defines a (binary) relation over the universe of the input word structure.
We call such a query \emph{relational}. We will also consider  \emph{functional} queries mapping tuples of elements to elements.

Another dynamic problem considered here is $\Member{L}$ which is defined similarly as $\RangeMember{L}$ but instead of $\qrange$ only has the Boolean query operation $\qmember$ that yields true if $w_0 \circ \ldots \circ w_{n-1} \in L$ holds.

    \section{Work-sensitive Dynamic Complexity}\label{section:model}
    \newcommand{\Rep}{U}

Since we are interested in the work that a dynamic program does, our specification mechanism for dynamic programs is considerably more elaborated than the one used in previous papers on dynamic complexity.
We introduce the mechanism in this section in two steps. First  the general form of dynamic programs and then  a more pseudo-code oriented syntax. 
Afterwards, we discuss how these dynamic programs translate into work-efficient constant-time parallel programs.

\subsection{The Dynamic Complexity Framework}
\label{section:model:DynFO}

Our  general form of dynamic programs mainly follows \cite{SchwentickZ16}, but is adapted to the slightly broader view of a dynamic problem as a data type.
For a more gentle introduction to dynamic complexity, we refer to \cite{SchwentickVZ20}.

The goal of a \emph{dynamic program} for a dynamic problem $\prob$ is to support all its operations~$\Delta$. To do so, it stores and updates an auxiliary structure $\aux$ over some schema $\auxSchema$, over the same domain as the input structure $\inp$ for $\prob$. 

 
A (first-order) dynamic program $\prog$ consists of a set of (first-order) \emph{update rules} for change operations and \emph{query rules} for query operations.\tsm{I stick to \emph{query rule}, here. As for change rules, a query rule uses a formula (or a program, for that matter).} More precisely, a program has one query rule over schema $\auxSchema$ per query operation that specifies how the (relational) result of that operation is obtained from the auxiliary structure. 
Furthermore, for each change operation $\delta \in \Delta$, it has one update rule per auxiliary relation or function that specifies the updates after a change based~on~$\delta$. 

A query rule is of the form
$\queryRule{\query(\tpl{p})}
  {\varphi_\query(\tpl{p})},
$
where $\varphi_\query$ is the (first-order) \emph{query formula} with free variables from $\tpl{p}$.

An update rule for a $k$-ary auxiliary relation $R$ is of the form
\[
    \changeRuleWhere{\delta(\tpl{p})}
        {R}{(t_1(\tpl{p};\tpl{x}), \ldots, t_k(\tpl{p};\tpl{x}))}
        {\varphi_\delta^R(\tpl{p};\tpl{x})}
       {C(\tpl{x})}.
\]
Here, $\varphi^R_\delta$ is the (first-order) \emph{update formula},
$t_1, \ldots, t_k$ are first-order terms (possibly using the $\mtext{ITE}$ construct) over $\auxSchema$, and
$C(\tpl{x})$, called a \emph{constraint} for the tuple \mbox{$\tpl x = x_1, \ldots, x_\ell$} of variables, is a conjunction of inequalities $x_i \le f_i(n)$ using functions $f_i \colon \N \to \N$, where $n$ is the size of the domain and $1 \le i \le \ell$.
We demand that all functions $f_i$ are first-order definable from $+$ and $\times$. 

The effect of such an update rule after a change operation $\delta(\tpl a)$ is as follows:
the new relation $R^{\calA'}$ in the updated auxiliary structure $\calA'$ contains all tuples from $R^\calA$ that are \emph{not} equal to $(t_1(\tpl{a};\tpl{b}), \ldots, t_k(\tpl{a};\tpl{b}))$ for any tuple $\tpl{b}$ that satisfies the constraints $C$; and additionally $R^{\calA'}$ contains all tuples $(t_1(\tpl{a};\tpl{b}), \ldots, t_k(\tpl{a};\tpl{b}))$ such that $\tpl b$ satisfies $C$ and $\aux \models \varphi_\delta^R(\tpl a; \tpl b)$ holds.

Phrased more operationally, an update is performed by enumerating all tuples $\tpl b$ that satisfy $C$, evaluating $\varphi_\delta^R(\tpl a; \tpl b)$ on the old auxiliary structure $\aux$, and depending on the result adding the tuple $(t_1(\tpl{a};\tpl{b}), \ldots, t_k(\tpl{a};\tpl{b}))$ to $R$ (if it was not already present), or removing that tuple from $R$ (if it was present).

Update rules for auxiliary functions are similar, but instead of an update formula that decides whether a tuple of the form $(t_1(\tpl{a};\tpl{b}), \ldots, t_k(\tpl{a};\tpl{b}))$ is contained in the updated relation, it features an update term that determines the new function value for a function argument of the form $(t_1(\tpl{a};\tpl{b}), \ldots, t_k(\tpl{a};\tpl{b}))$. 

We say that $\prog$ is a  dynamic program for a dynamic problem $\prob$  if it supports all its operations and, in particular, always yields correct results for query operations. More precisely, if the result of applying a query operation after a sequence $\alpha$ of change operations on an initial structure  $\inp_0$ yields the same result as the evaluation of the query rule on the auxiliary structure that is obtained by applying the update rules corresponding to the change operations in $\alpha$ to an initial auxiliary structure $\aux_0$.
%
Here, an initial input structure $\inp_0$ over some domain $\domain$ is \emph{empty}, that is, it is a structure with empty relations and with all function values being undefined ($\bot$). 
The initial auxiliary structure $\aux_0$ is over the same domain $\domain$ as $\inp_0$ and is defined from $\inp_0$ by some \FO-definable initialization.

By $\DynFO$, we denote the class of all dynamic problems that have a dynamic program in the sense we just defined.\tsm{\DynFO or $\DynFOar$?}

\subsection{A syntax for work-efficient dynamic programs}

In this paper we are particularly interested in dynamic programs that require little work to update the auxiliary structure after every change operation
and to compute the result of a query operation.
However, since dynamic programs do not come with an execution model, there is no direct way to define, say, when a \DynFO-programs has polylogarithmic-work, syntactically.
\longversion{But it is even not clear how a \emph{semantic} definition could be obtained: the obvious approach to require that the program has an implementation as a parallel program that only needs  polylogarithmic-work does not work, since it is not clear how to define when  a parallel program \emph{implements} a given dynamic program. This difficulty occurs in particular if one wants to prove that some problem does \emph{not} have a work-efficient dynamic program.

  Since, we are not interested in lower bounds in this paper, we follow a pragmatic approach here.}
We define a pseudo-code-based syntax for \emph{update} and \emph{query procedures} that will be used in place of the update and query \emph{formulas} in rules of dynamic programs. This syntax has three important properties: (1) it is reasonably well readable (as opposed to strict first-order logic formulas), (2) it allows a straightforward translation of rules into proper \DynFO-programs, and (3) it allows to associate a ``work-bounding function'' to each rule and to translate it into a PRAM program with $\bigO(1)$ parallel time and work bounded by this function.    

The  syntax of the pseudo-code has similarities with Abstract State Machines~\cite{Boerger05} and the PRAM-syntax of \cite{JaJa1992}. For simplicity, we describe a minimal set of syntactic elements that suffice for the dynamic programs in this paper. We encourage readers to have a look at Section~\ref{section:datastructures} for examples of update rules with pseudo-code syntax. 
\longversion{
  However, we mention some extensions that could be considered for more complicated programs.
}

We only spell out a syntax for \emph{update procedures} that can be used in place of the update formula $\varphi_\delta^R(\tpl{p};\tpl{x})$ of an update rule
\[
    \changeRuleWhere{\delta(\tpl{p})}
        {R}{(t_1(\tpl{p};\tpl{x}), \ldots, t_k(\tpl{p};\tpl{x}))}
        {\varphi_\delta^R(\tpl{p};\tpl{x})}
        {C(\tpl{x})}.
\]
Query procedures are defined similarly, but they can not invoke any change operations for supplementary instances, and their only free variables are from $\tpl p$.

We allow some compositionality: a dynamic program on some \emph{main instance} can use  \emph{supplementary instances} of other dynamic problems  and  invoke change or query operations of other dynamic programs on those instances. These supplementary instances are declared on a global level\tzm{Is this specified somewhere?} of the dynamic program and each has an associated identifier. 

Update procedures $P=P_1; P_2$ consist of two parts.  In the \emph{initial procedure} $P_1$ no reference to the free variables from $\tpl{x}$ are allowed, 
but change operations for 
supplementary instances can be invoked. We require that, for each change operation $\delta$ of the main instance 
and each supplementary  instance $\calS$, at most one update rule for $\delta$ invokes change operations for $\calS$.
\longversion{In general, some more flexibility could be added, e.g., additional local instances.\tzm{Not sure whether at this point this is helpful for readers. Maybe as footnote?}}
\tzm{It might be helpful to have a small example of an update procedure that shows how this might look like. Does not have to be meaningful. Or maybe refer to a later algorithm where many parts are used? And describe here which lines belong to initial and main procedure, which lines are instance invocations etc...}

In the \emph{main procedure} $P_2$, no change operations for supplementary instances  can be invoked, but references to $\tpl{x}$ are allowed.

More precisely, both $P_1$ and $P_2$ can use (a series of) instructions of the following forms:

\begin{itemize}
    \item assignments $f(\tpl y) \gets \term$ of a function value, 
    \item assignments $R(\tpl y) \gets \cond$ of a Boolean value, 
    \item conditional branches \ifte{\cond}{$P'$}{$P''$}, and
    \item parallel branches \fordo{$z \leq g(n)$}{$P'$}.
\end{itemize}

    Semantically, here and in the following $n$ always refers to the size of the domain of the main instance.
     The initial procedure $P_1$ can further use change invocations $\instance.\delta(\tpl y)$.
     \longversion{
    \begin{itemize}
   \item change invocations $\instance.\delta(\tpl y)$.
   \end{itemize}
   }
   However, they  are not allowed in the scope of parallel branches. And we recall that in $P_1$ no variables from $\tpl{x}$ can be used.
   
       The main procedure $P_2$ can further use return statements \ret{$\cond$} or \ret{$\term$}, but not inside parallel branches.
     \longversion{
   \begin{itemize}
    \item return statements \ret{$\cond$} or \ret{$\term$}, but not inside parallel branches.
    \end{itemize}
  }
  
    Of course, initial procedures can only have initial procedures $P'$ and $P''$ in conditional and parallel branches, and analogously for main procedures. 
    
Conditions and terms are defined as follows.
In all cases, $\tpl y$ denotes a tuple of terms and $z$ is a \emph{local variable}, not occurring in $\tpl{p}$ or $\tpl{x}$. 
In general, a \emph{term} evaluates to a domain element (or to $\bot$). It is built from
\begin{itemize}
    \item  local variables and variables from $\tpl p$ and $\tpl x$,
   \item  function symbols from $\auxSchema$ and previous function assignments,
     \item if-then-else terms \ifte{\cond}{\term$'$}{\term$''$},
    \item functional queries $\instance.\query(\tpl y)$, and
    \item expressions \algunique{$z \leq g(n)$}{\cond}. 
\end{itemize}
For the latter expression it is required that there is always exactly one domain element $a \le g(n)$ satisfying \cond.  

A \emph{condition} evaluates to $\true$ or $\false$. It may be
\begin{itemize}
    \item an atomic formula with relation symbols from $\auxSchema$ or previous assignments, with terms as above, 
    \item an expression \algexists{$z \leq g(n)$}{\cond}, 
    \item a relational query $\instance.\query(\tpl y)$ with terms $\tpl y$, and
    \item a Boolean combination of conditions. 
\end{itemize}

All functions $g \colon \N \to \N$ in these definitions are required to be \FO-definable.
For assignments of relations $R$ and functions $f$ we demand that these symbols do \emph{not} appear in $\auxSchema$.
If an assignment with a head $f(\tpl y)$ or $R(\tpl y)$ occurs in the scope of a parallel branch that binds variable $z$, then $z$ has to occur as a term $y_i$ in $\tpl y$.
We further demand that update procedures are well-formed, in the sense that
every execution path ends with a return statement of appropriate type.

In our pseudo-code algorithms, we display  update procedures $P=P_1;P_2$ with initial procedure $P_1$ and main procedure $P_2$ as
\begin{algorithmic}[0]
  \OnChangeWith{$\delta(\tpl p)$}{$P_1$}
        \UpdateAtWhereBy{$R$}{$(t_1(\tpl p, \tpl x), \ldots, t_k(\tpl p, \tpl x))$}{$C(\tpl x)$}
            $P_2$.
        \EndUpdateAtWhereBy
    \EndOnChangeWith
\end{algorithmic}
to emphasise that $P_1$ only needs to be evaluated once for the update of $R$, and not once for every different value of $\tpl x$.


In a nutshell, the semantics of an update rule
\[
    \changeRuleWhere{\delta(\tpl{p})}
        {R}{(t_1(\tpl{p};\tpl{x}), \ldots, t_k(\tpl{p};\tpl{x}))}
        {P}
        {C(\tpl{x})}
      \]
  is defined as in Subsection~\ref{section:model:DynFO}, but 
      $\aux \models \varphi_\delta^R(\tpl a, \tpl b)$ has to be replaced by the condition that $P$ returns true under the assignment $(\tpl p\mapsto \tpl a;\tpl x\mapsto \tpl b)$.



For update rules for auxiliary functions, $P$ returns the new function value instead of a Boolean value.

Since $P_1$ is independent of $\tpl x$, in the semantics, it is only evaluated once. In particular, any change invocations are triggered only once.





\newcommand{\PseudoDynFO}{\ensuremath{\text{Procedural-DynFO}}}
We refer to the above class of dynamic update programs as \PseudoDynFO-programs.
Here and later we will introduce abbreviations as syntactic sugar, for example the sequential loop \algfor{$z \leq m$}{$P$}, where $m \in \N$ needs to be a fixed natural number.

We show next that update and query procedures can be translated into constant-time CRAM programs. Since the latter can be translated into \FO-formulas \cite[Theorem 5.2]{immermanDC}, therefore \PseudoDynFO-programs can be translated in \DynFO-programs.

\longversion{
It is not hard to see that \PseudoDynFO-programs can be transformed into \DynFO-programs, as stated in the following proposition.

\begin{proposition}\label{prop:PseudoDynFOvsDynFO}
    If a dynamic problem $\prob$ can be specified by a \PseudoDynFO-program then $\prob \in \DynFO$.
  \end{proposition}
However, we do not need to prove this proposition, since we show next that \PseudoDynFO-programs can be translated into constant-time CRAMs and it is known that such programs can be translated into \FOar-formulas \cite[Theorem 5.2]{immermanDC}.
}
\subsection{Implementing \PseudoDynFO-programs as PRAMs}

We use \emph{Parallel Random Access Machines} (PRAMs) as the computational model
to measure the work of our dynamic programs.
A PRAM consists of a number of processors that work in parallel and use a shared memory.
We only consider \emph{CRAMs}, a special case  of Concurrent-Read Concurrent-Write model (CRCW PRAM),
i.e. processors are allowed to read and write concurrently from and to the same memory location,
but if multiple processors concurrently write the same memory location, then all of them need to write the same value.
For an input of size $n$ we denote the \emph{time} that a PRAM algorithm needs to compute the solution as $T(n)$.
The \emph{work} $W(n)$ of a PRAM algorithm is the sum of the number of all computation steps of all processors made during the computation.
For further details we refer to \cite{immermanDC,JaJa1992}.

It is easy to see that \PseudoDynFO{} programs $\calP$ can be translated into $\bigO(1)$-time CRAM-programs $\calC$.
To be able to make a statement about (an upper bound of) the work of $\calC$, we associate a function $\workbound$ with update rules and
show that every update rule $\pi$ can be implemented by a $\bigO(1)$-time CRAM-program with work $\bigO(\workbound)$.
Likewise for query rules.

In a nutshell, the work of an update procedure mainly depends on the scopes of the (nested) parallel branches and the amount of work needed to query and update the supplementary instances.
The work of a whole update rule is then determined by adding the work of the initial procedure once and adding the work of the main procedure for each tuple that satisfies the constraint of the update rule.

The function $\workbound$ is defined as follows. 
Let the update rule $\pi$ be of the form 
\begin{algorithmic}[0]
    \OnChange{$\delta(\tpl p)$}
        \UpdateAtWhereBy{$R$}{$(t_1(\tpl p, \tpl x), \ldots, t_k(\tpl p, \tpl x))$}{$x_1 \le g_1(n) \land \ldots \land x_\ell \le g_\ell(n)$}
            $P$
        \EndUpdateAtWhereBy
    \EndOnChange
  \end{algorithmic}
  with $P=P_1;P_2$ consisting of an initial procedure $P_1$ and a main procedure $P_2$.
  For simplicity we require that for each variable $x_i$ there is an inequality $x_i\le g_i(n)$, but it could be $g_i(n)=n$.
We set $\workbound(\pi) \df \max(\workbound(P_1), g_1(n) \cdot \ldots \cdot g_\ell(n) \cdot \workbound(P_2))$, where $\workbound(P)$   is inductively defined as follows.

For terms $t$ and conditions $C$,  we define $\workbound(t)$ and $\workbound(C)$, inductively. 

\[
    \workbound(t) \df \begin{cases}
    		1 &\text{if $t$ is a constant or a variable}\\
        \max(\workbound(t_1), \ldots, \workbound(t_\ell))  &\text{if $t = f(t_1, \ldots, t_\ell)$,}\\
        \max(\workbound(C), \workbound(t_1), \workbound(t_2))         &\text{if $t = \ifte{C}{t_1}{t_2}$}\\
        g(n) \cdot \workbound(C)                  &\text{if $t = \algunique{z \leq g(n)}{C}$.}
    \end{cases}
\]

\[
    \workbound(C) \df \begin{cases}
      1                                           &\text{if $C$ is atomic,}\\
      \max(\workbound(C_1),\ldots, \workbound(C_m)) &\text{if $C$ is a Boolean combination} \\ & \quad \text{of conditions } C_1,\ldots,C_m,\\
      g(n) \cdot \workbound(C')                  &\text{if $C = \algexists{z \leq g(n)}{C'}$},\\
      \workbound(\pi') & \text{if $C$ is $\instance.\query(\tpl y)$ and $\pi'$ is} \\ & \quad \text{its query rule.}
    \end{cases}
  \]
Furthermore,
\begin{itemize}
    \item $\workbound(f(\tpl y) \gets C) \df \workbound(C)$,
    \item $\workbound(R(\tpl y) \gets t) \df \workbound(t)$,
  \item $\workbound(\ifte{C}{P'}{P''}) \df \max(\workbound(C),\workbound(P'),\workbound(P''))$,
    \item $\workbound(\fordo{z \leq g(n)}{P'}) \df g(n) \cdot \workbound(P')$,
    \item $\workbound(\ret{\varphi}) \df 1$,
    \item $\workbound(\instance.\delta(\tpl y)) \df \max(\workbound(\pi_1),\ldots, \workbound(\pi_m))$, where $\pi_1,\ldots,\pi_m$ are the update rules for change operation $\delta$.
\end{itemize}    

\begin{proposition}\label{prop:PseudoDynFOvsPRAM}
    For every update rule $\pi$ a $\bigO(1)$-time PRAM-program  with work $\bigO(\workbound(\pi))$ can be constructed. Likewise for query rules.
  \end{proposition}
  A sketch for the straightforward proof can be found in the appendix. 


\tsm{Proof sketch moved to appendix}

It follows from Proposition~\ref{prop:PseudoDynFOvsPRAM} that a dynamic program $\prog$ can be implemented by a $\bigO(1)$-time PRAM-program with work $\bigO(f)$, if for each update rule $\pi$ of $\prog$ it holds $\workbound(\pi)\le f$.
\tsm{This is \emph{not} the same as picking a maximum $\pi$ since that might not exist.}


    \section{A simple work-efficient Dynamic Program}\label{section:datastructures}
    \newcommand{\fst}{\text{1st}}
\newcommand{\snd}{\text{2nd}}
\newcommand{\anc}{\text{anc}}
\newcommand{\leaf}{\text{leaf}}
\newcommand{\rootT}{t}
\newcommand{\rootE}{\text{root}}

In this section we consider a simple dynamic problem with a fairly
work-efficient dynamic program. It serves as an example for our framework but will also be used as a subroutine in later sections.

The dynamic problem is to maintain a subset $K$ of an ordered set $D$ of elements under insertion and removal of elements in $K$, allowing for navigation from an element of $D$ to the next larger and
smaller element in~$K$. That is, we consider the following dynamic problem:

\newcommand{\NextInK}{\myproblem{NextInK}} 
\dynProblemDefinition{\NextInK}
    {A set $K \subseteq D$ with canonical linear order $\le$ on $D$}
    {\begin{qclist}
        \qclitem{$\ins(i)$:}{}{Inserts $i \in D$ into $K$}
        \qclitem{$\del(i)$:}{}{Deletes $i \in D$ from $K$}
    \end{qclist}}
    {\begin{qclist}
        \qclitem{$\mtext{pred}(i)$:}{}{Returns predecessor of $i$ in $K$, that is, $\max\{j \in K \mid i > j\}$}
        \qclitem{$\mtext{succ}(i)$:}{}{Returns successor of $i$ in $K$, that is, $\min\{j \in K \mid i < j\}$}
      \end{qclist}}

    For the smallest (largest) element the result of a $\mtext{pred}$ ($\mtext{succ}$) query is undefined, i.e. $\bot$.
    For simplicity, we assume in the following that $D$ is always of the form $[n]$, for some~$n\in\N$.

Sequentially, the changes and queries of \NextInK can be handled in sequential time $\bigO(\log \log n)$ \cite{FrandsenMiltersen+1997}.
Here we show that the problem also has a dynamic program with parallel time $\bigO(1)$ and work $\bigO(\log n)$. 
\begin{lemma}\label{lem:NextInKWork}
  There is a \DynFO-program for  \NextInK  with $\bigO(\log n)$ work per change and query operation.
\end{lemma}
\begin{proof}
  The dynamic program uses an ordered binary balanced tree $T$ with leave set  $[n]$,
  and with $0$ as its leftmost leaf. Each inner node $v$ represents the interval $S(v)$ of numbers labelling the leafs of the subtree of $v$. 
To traverse the tree, the dynamic program uses functions $\fst$ and $\snd$ that map
an inner node to its first or second child, respectively, and a function $\anc(v,j)$ that
returns the $j$-th ancestor of $v$ in the tree.
\longversion{
  Formally, the $2|D|$ nodes of $T$ can be represented by pairs $(a,b)$ of elements from the domain $D$, where $b$ indicates the height of the node in the tree, and  $a$ gives the left-to-right position of the node among all nodes with the same height.
Then, we have to use pairs of binary functions $\fst^1, \fst^2$ etc.\ that map encodings of a node to the components of the encoding of another node. We disregard these technical issues and use nodes of $T$ just as domain elements, and also identify an element $i \in [n]$ with the $i$-th leaf of $T$. 
The technical translation is straightforward and does not affect our
reasoning regarding \FO-expressibility and work of dynamic programs.
}
So, $\anc(v,2)$
returns the parent of the parent of $v$.
\longversion{
  If there is no $j$-th ancestor of a node $v$, then $\anc(v,j)$ is undefined.   
  }
  
%


The functions $\fst$, $\snd$ and $\anc$ are static, that is, they are initialized beforehand and not affected by change operations. 

The idea of the dynamic program is to maintain, for each node $v$, the maximal and minimal element in $K \cap S(v)$ (which is undefined if
    $K\cap S(v)=\emptyset$), by maintaining two functions $\min$ and $\max$. It is easy to see that this information can be updated and  queries be answered in $\bigO(\log n)$ time as the tree has depth~$\bigO(\log n)$. For achieving $\bigO(\log n)$ work and constant time, we need to have a closer look.

    
    Using $\min$ and $\max$, it is easy to determine the $K$-successor
    of an element~\mbox{$i\in D$}: if $v$ is the lowest ancestor of $i$ with
    $\max(v) > i$, then the $K$-successor of $i$ is $\min(w)$ for the second child $w \df \snd(v)$ of $v$. Algorithm~\ref{alg:NextInK:Succ} shows a query rule for the
    query operation $\mtext{succ}(i)$. 
    \begin{algorithm}[t]
        \begin{algorithmic}[1]
            \OnQuery{$\mtext{succ}(i)$}
                    \If{$\max(T.\rootE) \leq i$}
                        \State \Return $\bot$
                    \Else\
                        \State $k \gets \algunique{1 \le k \le \log(n)}{\max(T.\anc(i,k)) > i}$\\ \hspace{6cm}$ \land \; \max(T.\anc(i,k-1)) \leq i$  \label{alg:NextInK:Succ:lineUnique}
                        \State \Return $\min(T.\snd(T.\anc(i,k)))$
                    \EndIf
            \EndOnQuery
        \end{algorithmic}
        \caption{Querying a successor.}
        \label{alg:NextInK:Succ}
    \end{algorithm}
    The update of these functions is easy when an element $i$ is
    inserted into $K$. This is spelled out for $\min$  in Algorithm~\ref{alg:NextInK:UpdateMinAfterIns}. The dynamic program only needs to check if the
    new element becomes the minimal element in $S(v)$, for
    every node $v$ that is an ancestor of the leaf $i$.

    \begin{algorithm}[t]
        \begin{algorithmic}[1]
            \OnChangeAtWhereBy{$\ins(i)$}{$\min$}{$T.\anc(i,k)$}{$k \le \log n$}
                    \State $v \gets T.\anc(i,k)$
                    \If{$\min(v)>i$}
                        \State \Return $i$
                    \Else\
                        \State \Return $\min(v)$
                    \EndIf
            \EndOnChangeAtWhereBy
        \end{algorithmic}
        \caption{Updating $\min$ after an insertion.}
        \label{alg:NextInK:UpdateMinAfterIns}
    \end{algorithm}

  Algorithm~\ref{alg:NextInK:UpdateMinAfterDel} shows how $\min$ can
  be updated if an element $i$ is deleted from~$K$: if $i$
  is the minimal element of $K$ in $S(v)$, for some node
  $v$, then $\min(v)$ needs to be replaced by its $K$-successor,
  assuming it is in $S(v)$. 

    \begin{algorithm}[t]
        \begin{algorithmic}[1]
            \OnChange{$\del(i)$}
			\With{}
			 	\State $s \gets \mtext{succ}(i)$ \label{alg:NextInK:UpdateMinAfterDel:lineSucc}
			\EndWith            
            \UpdateAtWhereBy{$\min$}{$T.\anc(i,k)$}{$k\le \log n$}
            
                    \State $v \gets T.\anc(i,k)$
                    \If{$\min(v) \neq i$}
                        \State \Return $\min(v)$
                    \ElsIf{$\max(v) = i$}
                        \State \Return $\bot$
                    \Else
                        \State \Return $s$
                    \EndIf
            \EndUpdateAtWhereBy
            \EndOnChange        
        \end{algorithmic}
        \caption{Updating $\min$ after a deletion.}
        \label{alg:NextInK:UpdateMinAfterDel}
    \end{algorithm}

   It is easy to verify the claimed  work upper bounds for $\prog$.
    Querying a successor or predecessor via Algorithm~\ref{alg:NextInK:Succ} needs $\bigO(\log n)$ work,
    since Line \ref{alg:NextInK:Succ:lineUnique} requires
    $\bigO(\log n)$ and all others require $\bigO(1)$ work.
    For maintaining the function $\min$ the programs in Algorithms \ref{alg:NextInK:UpdateMinAfterIns} and~\ref{alg:NextInK:UpdateMinAfterDel}
    update the value of $\log n$ tuples, but the work per tuple is
    constant.     In the case of a deletion, Line
    \ref{alg:NextInK:UpdateMinAfterDel:lineSucc} requires
    $\bigO(\log n)$ work but is executed only once. The remaining
    part consists of $\bigO(\log n)$ parallel executions of
    statements, each
    with $\bigO(1)$ work.

    The handling of $\max$ and its work analysis is analogous.
\end{proof}



    \section{Regular Languages}\label{section:regular}
    In this section, we show that the range problem can be maintained with
$o(n)$ work  for all regular languages and with polylogarithmic work
for star-free languages.
For the former we show how to reduce the work of a known \DynFO-program.
For the latter we translate the idea of \cite{FrandsenMiltersen+1997} for maintaining the range problem for star-free languages in $\bigO(\log \log n)$ sequential time
into a dynamic program with $\bigO(1)$ parallel time.

\subsection{DynFO-programs with sublinear work for regular languages}
\begin{theorem}\label{theo:regularWorkBound}
    Let $L$ be a regular language. Then $\RangeMember{L}$ can be maintained in DynFO with work $\bigO(n^{\epsilon})$ per query and change operation, for every $\epsilon > 0$.
\end{theorem}

The proof of this theorem makes use of the algebraic view of
regular languages. For readers not familiar with this view, the basic
idea is as follows: for a fixed DFA $\calA=(Q,\Sigma,\delta,q_0,F)$, we
first associate with each string $w$ a
function $f_w$ on $Q$ that is induced by the behaviour of $\calA$ on $w$
via $f_w(q)\df \delta^*(q,w)$, where $\delta^*$ is the extension of the transition function $\delta$ to strings. The set of all  functions $f \colon Q \to Q$
with composition as binary operation is a \emph{monoid}, that is, a
structure with an associative binary operation $\circ$ and a neutral element,
the identity function. Thus, composing the effect of $\calA$ on
subsequent substrings of a string corresponds to multiplication of the
monoid elements associated with these substrings. The \emph{syntactic monoid} $M(L)$ of a regular language $L$ is basically the monoid associated
with its minimal automaton.

It is thus clear that, for the dynamic problem
$\RangeMember{L}$ where $L$ is regular, a dynamic program  can be easily obtained from a
dynamic program for the dynamic problem $\RangeEval{M(L)}$, where
$\RangeEval{M}$, for finite monoids~$M$, is defined as
follows.\footnote{We note that, unlike for words, each position always
  carries a monoid element. However, the empty string of the word case
corresponds to the neutral element in the monoid case. In particular,
the initial ``empty'' sequence consists of $n$ copies of the neutral element.}

\dynProblemDefinition{$\RangeEval{M}$}
    {A sequence $m_0 \ldots m_{n-1}$ of monoid elements $m_i \in M$}
    {\begin{qclist}
        \qclitem{$\set_m(i)$ for $m \in M$:}{}{Replaces $m_i$ by $m$}
    \end{qclist}}
    {\begin{qclist}
        \qclitem{$\qrange(\ell,r)$:}{}{$m_\ell \circ \cdots \circ m_r$}
    \end{qclist}}



For the proof of Theorem~\ref{theo:regularWorkBound} we do not need any 
insights into monoid theory. However, when studying languages definable by first-order formulas in Theorem~\ref{theo:FO-work}
below, we will make use of a known decomposition result.

From the discussion above it is now clear that in order to prove
Theorem~\ref{theo:regularWorkBound}, it suffices to prove the
following result. 

\begin{proposition}\label{prop:monoidWorkBound}
    Let $M$ be a finite monoid. For every $\epsilon > 0$, $\RangeEval{M}$ can be maintained in DynFO with work $\bigO(n^{\epsilon})$ per query and change operation.
  \end{proposition}

\begin{proofsketch}
  In \cite{GeladeMS12}, it was (implicitly) shown that
  $\RangeMember{L}$ is in \DynProp (that is, quantifier-free \DynFO), for regular languages $L$. The idea was
    to maintain the effect of a DFA for $L$ on  $w[\ell,r]$, for
    each interval $(\ell,r)$ of positions. This approach can be easily
    used for $\RangeEval{M}$ as well, but it requires a quadratic
    number of updates after a change operation, in the worst case.

We adapt this approach and only store the effect of the DFA for $\bigO(n^\epsilon)$ intervals, by considering a hierarchy of intervals of bounded depth.

    The first level in the hierarchy of intervals is obtained by decomposing the input sequence into intervals of length $t$, for a carefully chosen $t$.
    We call these intervals \emph{base intervals} of height~$1$ and their subintervals \emph{special intervals} of height~$1$.
    The latter are \emph{special} in the sense that they are exactly the intervals for which the dynamic program maintaines the product of monoid elements.
    In particular, each base interval of height $1$ gives rise to $\bigO(t^2)$ special intervals of height $1$.
    The second level of the hierarchy is obtained by decomposing the sequence of base intervals of height $1$ into sequences of length~$t$.
    Each such sequence of length $t$ is combined to one base interval of height $2$; and each contiguous subsequence of such a sequence is combined to one special interval of height~$2$. Again, each base interval of height $2$ gives rise to $\bigO(t^2)$ special intervals of height~$2$.  This process is continued recursively for the higher levels of the hierarchy, until only one base interval of height $h$ remains. We refer to Figure \ref{fig:regular-example} for an illustration of this construction.
    
    The splitting factor $t$ is chosen in dependence of $n$ and $\epsilon$ such that the height of this hierarchy of special intervals only depends on $\epsilon$ and is thus constant for all $n$.
    More precisely, we fix $\lambda \df \frac{\epsilon}{2}$ and $t \df n^{\lambda}$, rounded up.
    \longversion{For  simplicity of exposition, we assume from now on that $n=t^{\frac{1}{\lambda}}$.}
    Therefore,  $h=\log_t(n) = \frac{1}{\lambda}$. 

    The idea for the dynamic program is to store the product of monoid elements for each special interval. The two crucial observations are then,
    that (1) the product of each (not necessary special) interval can be computed with the help of a constant number of special intervals, and
    (2) that each change operation affects at most $t^2$ special
    intervals per level of the hierarchy and thus at most $ht^2 \in
    \bigO(n^\epsilon)$ special intervals in total. We refer to the appendix for more details.  
%
    \begin{figure}[t]
        \begin{center}
           \scalebox{0.85}{
  \begin{tikzpicture}[scale=0.05,thick]
    \newcommand{\varn}{27}
    \newcommand{\vart}{3}
    \newcommand{\varh}{3}

    \newcommand{\borderoffset}{4}
    \newcommand{\elementdistance}{9}
    \newcommand{\leveldistance}{16}
    \newcommand{\heightBase}{6}
    \newcommand{\heightOffsetPerSpecial}{2}

    \newcommand{\nohighlightcolor}{black}
    \newcommand{\querycolor}{blue}
    \newcommand{\changecolor}{red}

    \newcommand{\intervalcolor}{\nohighlightcolor}

    \pgfmathtruncatemacro{\varnMinusOne}{\varn-1}
    \foreach \i in {0,...,\varnMinusOne}
    {
        \node[draw=none] at (\elementdistance*\i,-1) {\tiny $m_{\i}$};
    }
    
    \foreach \varlevel in {1, ..., \varh}
    {
        \pgfmathtruncatemacro{\height}{-1+(\varlevel)*\leveldistance}
        \pgfmathtruncatemacro{\heightBaseUpper}{\height + \heightBase/2}
        \pgfmathtruncatemacro{\heightBaseLower}{\height - \heightBase/2}

        \pgfmathtruncatemacro{\countBaseIntervals}{\varn / (\vart^\varlevel)-1}

        \node[draw=none] at (-22, \heightBaseLower) {level $\varlevel$};

        \edef\specialIntervalLinesCount{3}
        \pgfmathtruncatemacro{\levelRectTop}{\height + \heightOffsetPerSpecial*1.5}
        \pgfmathtruncatemacro{\levelRectBottom}{\height - \heightOffsetPerSpecial * (\specialIntervalLinesCount+1.5)}
        \pgfmathtruncatemacro{\levelRectLeft}{-\elementdistance}
        \pgfmathtruncatemacro{\levelRectRight}{\varn * \elementdistance -1}
        \definecolor{levelbackgroundcolor}{rgb}{0.93,0.93,0.93}
        \fill [levelbackgroundcolor,rounded corners=2pt] (\levelRectLeft,\levelRectBottom) rectangle (\levelRectRight, \levelRectTop);

        \pgfmathtruncatemacro{\elementCountBase}{(\vart^\varlevel)}
        \pgfmathtruncatemacro{\lengthBase}{(\elementdistance * \elementCountBase - \borderoffset}
        \pgfmathtruncatemacro{\lengthBaseLastWithoutOffset}{(\elementdistance * \vart^(\varlevel-1))}

        \foreach \i in {0, ..., \countBaseIntervals}
        {
            \pgfmathtruncatemacro{\startBase}{\elementdistance*\elementCountBase*\i - \borderoffset + 1}

            \ifnum \varlevel=1
                \ifnum \i=7
                    \renewcommand{\intervalcolor}{\changecolor}
                \fi
            \else
                \ifnum \varlevel=2
                    \ifnum \i=2
                        \renewcommand{\intervalcolor}{\changecolor}
                    \fi
                \else
                    \renewcommand{\intervalcolor}{\changecolor}
                \fi
            \fi

            \draw[draw=\intervalcolor] (\startBase,\height) -- +(\lengthBase,0);
            \renewcommand{\intervalcolor}{black}

            \pgfmathtruncatemacro{\tMinusTwo}{\vart-2}
            \edef\specialIndex{1}
            \foreach \specialstartelement in {0, ..., \tMinusTwo}
            {
                \pgfmathtruncatemacro{\heightSpecial}{\height - \heightOffsetPerSpecial * \specialIndex}
                \pgfmathtruncatemacro{\startSpecial}{\startBase + \specialstartelement * \lengthBaseLastWithoutOffset}

                \pgfmathtruncatemacro{\tMinusL}{\vart-\specialstartelement}
                \ifnum \specialstartelement=0
                    \pgfmathparse{\tMinusL-1}
                    \xdef\tMinusL{\pgfmathresult}
                \else
                \fi
                \foreach \specialwidth in {2, ..., \tMinusL}
                {
                    \pgfmathtruncatemacro{\lengthSpecial}{\specialwidth * \lengthBaseLastWithoutOffset - \borderoffset}
                    
                    \ifnum \varlevel=1
                        \ifnum \i=7
                            \ifnum \specialstartelement=0
                                \renewcommand{\intervalcolor}{\querycolor}
                            \fi
                        \fi
                    \fi

                    \ifnum \varlevel=2
                        \ifnum \i=0
                            \ifnum \specialstartelement=1
                                \renewcommand{\intervalcolor}{\querycolor}
                            \fi
                        \fi
                    \fi

                    \draw[draw=\intervalcolor] (\startSpecial,\heightSpecial) -- +(\lengthSpecial,0);
                    \renewcommand{\intervalcolor}{black}
                    
                    \pgfmathparse{\specialIndex+1}
                    \xdef\specialIndex{\pgfmathresult}
                }

            }

            \edef\specialwidth{1}
            \pgfmathtruncatemacro{\lengthSpecial}{\specialwidth * \lengthBaseLastWithoutOffset - \borderoffset}

            \pgfmathtruncatemacro{\tMinusOne}{\vart-1}
            \foreach \specialstartelement in {0, ..., \tMinusOne}
            {
                \pgfmathtruncatemacro{\heightSpecial}{\height - \heightOffsetPerSpecial * \specialIndex}
                \pgfmathtruncatemacro{\startSpecial}{\startBase + \specialstartelement * \lengthBaseLastWithoutOffset}

                \ifnum \varlevel=1
                    \ifnum \i=0
                        \ifnum \specialstartelement=2
                            \renewcommand{\intervalcolor}{\querycolor}
                        \fi
                    \fi
                \fi
                        
                \ifnum \varlevel=2
                    \ifnum \i=2
                        \ifnum \specialstartelement=0
                            \renewcommand{\intervalcolor}{\querycolor}
                        \fi
                    \fi
                \fi

                \ifnum \varlevel=3
                    \ifnum \specialstartelement=1
                        \renewcommand{\intervalcolor}{\querycolor}
                    \fi
                \fi

                \draw[draw=\intervalcolor] (\startSpecial,\heightSpecial) -- +(\lengthSpecial,0);
                \renewcommand{\intervalcolor}{black}
            }
        }

    }

\end{tikzpicture}
}
        \end{center}
        \caption{
            Illustration of special intervals, for $t=3$. 
            The special intervals of level $3$ are $[0,9), [9,18), [18,27), [0,18)$ and $[9,27)$ with base interval $[0,27)$. 
            The result of a query $\qrange(2,22)$ can be computed as
            \( \prod_{i=2}^{22} m_i = \big( m[2,3) \circ m[3,9) \big) \circ m[9,18)\circ \big( m[18,21) \circ m[21, 23) \big)\),
            illustrated above in {\color{blue}blue}.
            The affected base intervals for a change at position $23$ are marked in {\color{red}red}.
            E.g., the new product $m'[18,27)$ can be computed by $m'[18,27) = m[18,21)\circ m'[21,24) \circ m[24,27)$.
            As the products are recomputed bottom up, $m'[21,24)$ is already updated.
        }
        \label{fig:regular-example}
    \end{figure}
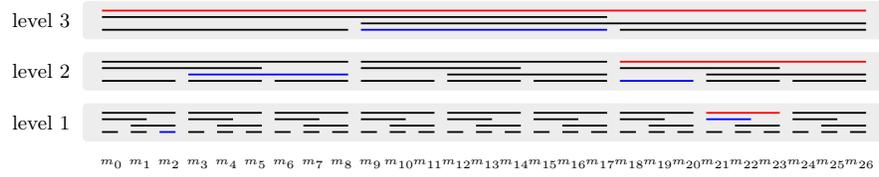
\end{proofsketch}

\subsection{DynFO-programs with polylogarithmic work for star-free languages}
Although the work bound of Theorem~\ref{theo:regularWorkBound}  for
regular languages is strongly sublinear, one might aim for an even
more work-efficient dynamic program, especially, since
$\RangeMember{L}$ can be maintained \emph{sequentially} with
logarithmic update time for regular languages
\cite{FrandsenMiltersen+1997}. We leave it as an open problem whether  for every regular
language $L$ 
there is a \DynFO-program for $\RangeMember{L}$ with a polylogarithmic work bound. 
However, we show next that such programs exist for star-free regular
languages, in fact they even have  a logarithmic  work bound. The
star-free languages are those that can be expressed by regular
expressions that do not use the Kleene star operator but can use
complementation. 

\begin{theorem}\label{theo:FO-work}
    Let $L$ be a star-free regular language. Then $\RangeMember{L}$ can be maintained in \DynFO with work $\bigO(\log n)$ per query and change operation.
  \end{theorem}
It is well-known that star-free regular languages are just the regular
languages that can be defined in first-order logic (without
arithmetic!) \cite{McNaughtonP71}. Readers might ask why we consider
dynamic first-order maintainability of a problem that can actually be
\emph{expressed} in first-order logic. The key point is the parallel work here:
even though the membership problem for star-free languages can be
solved by a parallel algorithm in time $\bigO(1)$, it inherently requires parallel work $\Omega(n)$. 
\begin{proofsketch}
    The proof uses the well-known connection between star-free
    languages and group-free monoids (see, e.g., \cite[Chapter
    V.3]{Straubing1994} and \cite[Theorem V.3.2]{Straubing1994}). It
   thus    follows the approach of \cite{FrandsenMiltersen+1997}.
    
     In a nutshell, our dynamic program simply implements the algorithms
  of the proof of Theorem 2.4.2 in \cite{FrandsenMiltersen+1997}. Those
  algorithms consist of a constantly bounded number of simple operations and a
 constantly bounded number of searches for a next neighbour in a set. Since
 the latter can be done in \DynFO with work $\bigO(\log n)$ thanks to
 Lemma~\ref{lem:NextInKWork}, we get the desired result for group-free
 monoids and then for star-free languages.
 We refer to the appendix for more details.
  \end{proofsketch}


    \section{Context Free Languages}\label{section:cfl}
    \newcommand{\str}{\textnormal{str}}
\newcommand{\open}{\langle}
\newcommand{\close}{\rangle}
\newcommand{\red}{\mu}

As we have seen in Section~\ref{section:regular}, range queries to regular languages can be maintained in \DynFO with strongly sublinear work. An immediate question is whether context-free languages are equally well-behaved. Already the initial paper by Patnaik and Immerman showed that \DynFO can maintain the membership problem for \emph{Dyck languages} $D_k$, for $k \geq 1$, that is, the languages of well-balanced parentheses expressions with $k$ types of parentheses \cite{PatnaikI94}. It was shown afterwards in \cite[Theorem 4.1]{GeladeMS12} that \DynFO actually captures the membership problem for all context-free languages and that Dyck languages even do not require quantifiers in formulas (but functions in the auxiliary structure) \cite[Proposition 4.4]{GeladeMS12}. These results can easily be seen to apply to range queries as well. However, the dynamic program of \cite[Theorem 4.1]{GeladeMS12} uses 4-ary relations and three nested existential quantifiers, yielding work in the order of $n^7$.

In the following, we show that the membership problem for context-free languages is likely  \emph{not} solvable in \DynFO with sublinear work, but that the Dyck language $D_1$ with one bracket type can be handled with polylogarithmic work for the membership problem and work $\bigO(n^\epsilon)$ for the range problem, and that for other Dyck languages these bounds hold with an additional linear factor $n$.

\subsection{A conditional lower bound for context-free languages}

Our conditional lower bound for context-free languages is based on a result from Abboud et al.\ \cite{AbboudBackurs+2018} and the simple  observation  that the word problem for a language $L$ can be solved, given a dynamic program for its  membership problem.
\begin{lemma}\label{lem:staticUsingDyn}
    Let $L$ be a language.
    If $\Member{L}$ can be maintained in $\DynFO$ with work $f(n)$, then the word problem for $L$ can be decided sequentially in  time $\bigO(n \cdot f(n))$.
\end{lemma}

The announced lower bound is relative to the following conjecture  \cite{AbboudBBK17}.
\begin{conjecture}[$k$-Clique conjecture]
  For any $\epsilon>0$, and $k\ge 3$, $k$-Clique has no algorithm with time bound $\bigO(n^{(1-\epsilon)\frac{\omega}{3}k})$.
\end{conjecture}
Here, $\omega$ is the matrix multiplication exponent \cite{Gall2014,Williams2012}, which is known to be smaller than $2.373$ and believed to be exactly two \cite{Gall2014,Williams2012}.

In \cite {AbboudBackurs+2018}, the word problem for context-free languages was linked to the $k$-Clique problem as follows. 
\begin{theorem}[{\cite[Theorem 1.1]{AbboudBackurs+2018}}]\label{theo:cfg-clique}
  There is a context free grammar $G$ such that, if the word problem for $L(G)$ can be solved in time $T(n)$, 
    $k$-Clique can be solved on $n$ node graphs in $\bigO(T(n^{\frac{k}{3}+1}))$ time, for any $k \ge 3$.
  \end{theorem}

Putting Lemma~\ref{lem:staticUsingDyn} and Theorem~\ref{theo:cfg-clique} together, we get the following result.
\begin{theorem}\label{theo:lower}
  There is a context free grammar $G$ such that, if the membership problem for $L(G)$ can be solved by a  \DynFO-program with work $\bigO(n^{\omega-1-\epsilon})$, for some $\epsilon>0$, then the $k$-Clique conjecture fails.
\end{theorem}
The simple proofs of Lemma~\ref{lem:staticUsingDyn} and Theorem~\ref{theo:lower} are presented in the appendix.

Thus, we can not reasonably expect any \DynFO-programs for general context-free languages with considerable less work than $\bigO(n^{1.37})$ barring any breakthroughs for matrix multiplication. In fact, for ``combinatorial \DynFO-programs'', an analogous reasoning yields a work lower bound of $\bigO(n^{2-\epsilon})$.

\subsection{On work-efficient dynamic programs for Dyck languages}

We next turn to Dyck languages. Clearly, all Dyck languages are deterministic context-free, their word problem can therefore be solved in linear time, and thus the lower bound approach of the previous subsection does not work for them.
In \cite{FrandsenHusfeldt+1995} it was shown that $D_1$ is maintainable by a sequential dynamic algorithm in time $\bigO(\log n)$,  per change and query operation, and $D_k$ in time $\bigO((\log n)^3 \cdot \log^\ast n)$, for every $k > 1$. We already mentioned that their membership problem can be maintained in \DynFO. The challenge is to find dynamic programs (with parallel time $\bigO(1)$) with little work.

We show that it is actually possibly by suitably adapting the idea from \cite{FrandsenHusfeldt+1995} to get a \DynFO-program for $D_1$ with polylogarithmic work, but for $D_k$, for $k>1$, the approach so far only yields  $\DynFO$-programs with an additional linear factor in the work bound.
We first present the \DynFO-program with polylogarithmic work for the membership problem of $D_1$. It basically mimics the sequential algorithm from  \cite{FrandsenHusfeldt+1995} that maintains $D_1$. 
\begin{theorem}\label{thm:D1polylog}
    $\Member{D_1}$ can be maintained in \DynFO with $\bigO((\log n)^3)$ work.
  \end{theorem}

%
%
%
%
%
\begin{proofsketch}
Let  $\Sigma_1 = \{\open,\close\}$ be the alphabet underlying $D_1$.
  The dynamic program uses an ordered binary tree $T$ such that each leaf corresponds to one position from left-to right. A parent node corresponds to the set of positions of its children. We assume for simplicity that the domain is $[n]$, for some number $n$ that is a power of 2.
If the input string is $w$ then each node $x$ of $T$ represents a substring $\str(x)$ of $w$ via the leaves of the induced subtree at $x$ in a natural fashion. 
The main idea from \cite{FrandsenHusfeldt+1995} is to  maintain for each node $x$ of $T$ information about the unmatched parentheses of $\str(x)$. The input word is in the Dyck language if there are no unmatched parentheses in the root of $T$.

  In a nutshell, the dynamic program for $\Member{D_1}$ maintains for each node $x$ of $T$ the numbers $\ell(x)$ and $r(x)$ that represent the number of
  unmatched closing and unmatched opening brackets of the string $\str(x)$. E.g.,  if that string is $\close\open\close\close\open\open\close$ for $x$, then $\ell(x)=2$ and $r(x)=1$. The overall string $w$ is in $D_1$ exactly if $r(\textnormal{root}) = \ell(\textnormal{root}) = 0$.

In the algorithm of \cite{FrandsenHusfeldt+1995}, the functions $\ell$ and $r$ are updated in a bottom-up fashion. However, we will observe that they do not need to be updated sequentially in that fashion, but can be updated  in parallel constant time.
   In the following, we describe how $\prog$ can update $\ell(x)$ and $r(x)$ for all ancestor nodes $x$ of a position~$p$,  after a closing parenthesis $\close$ was inserted at $p$.   Maintaining $\ell$ and $r$ for the other change operations is analogous.
    
    There are two types of effects that an insertion of a closing parenthesis could have on $x$:
    either $\ell(x)$ is increased by one and $r(x)$ remains unchanged, or $r(x)$ is decreased by one and $\ell(x)$ remains unchanged.
    We denote these effects by the pairs $(+1,0)$ and $(0,-1)$, respectively.

    Table~\ref{tbl:table-d1-cases} shows how the effect of a change at a position $p$ below a node $x$ with children $y_1$ and $y_2$ relates to the effect at the affected child. This depends on whether $r(y_1) \leq \ell(y_2)$ and whether the affected child is $y_1$ or $y_2$.
 \begin{table}[t]
        \begin{center}
            \begin{tabular}{l|c|c}
                                        & $p$ is in $\str(y_1)$                                                    & $p$ is in $\str(y_2)$ \\ \hline
                $r(y_1) \leq \ell(y_2)$ & $\begin{array}{c}(+1,0) \to (+1,0) \\ (0,-1) \to (+1,0)\end{array}$   & $\begin{array}{c}(+1,0) \to (+1,0) \\ (0,-1) \to (0,-1)\end{array}$ \\ \hline
                $r(y_1) > \ell(y_2)$ & $\begin{array}{c}(+1,0) \to (+1,0) \\ (0,-1) \to (0,-1)\end{array}$   & $\begin{array}{c}(+1,0) \to (0,-1) \\ (0,-1) \to (0,-1)\end{array}$ \\ 
            \end{tabular}
        \end{center}
        \caption{
            The effect on $x$ after a closing parenthesis was inserted at position $p$.
            The effects depend on the effect on the children $y_1$ and $y_2$ of $x$:
            for example, an entry '$(0,-1) \to (+1,0)$' in the column '$p$ is in $\str(y_1)$' means that
            if the change operation has effect $(0,-1)$ on $y_1$ then the change operation has effect $(+1,0)$ on $x$.
        }
        \label{tbl:table-d1-cases}
      \end{table}
       A closer inspection of Table~\ref{tbl:table-d1-cases} reveals a crucial observation: in the  upper left and the lower right field of the table, the effect on $x$ is \emph{independent} of the effect on the child (being it $y_1$ or $y_2$). That is, these cases induce an effect on $x$ independent of the children. We thus call these cases \emph{effect-inducing}. In the other two fields, the effect on $x$ depends on the effect at the child, but in the simplest possible way: they are just the same. That is the effect at the child is just adopted  by~$x$. We call these cases \emph{effect-preserving}. To determine the effect at $x$ it is thus sufficient to identify the highest affected descendant node~$z$ of~$x$, where an effect-inducing case applies, such that for all intermediate nodes between $x$ and $z$ only effect-preserving cases apply.

Our dynamic program implements this idea. First it determines, for each ancestor $x$ of the change position $p$, whether it is effect-inducing and which effect is induced. Then it identifies, for each $x$, the node $z$ (represented by its height $i$ above $p$) as the unique effect-inducing node that has no effect-inducing node on its path to $x$.
The node $z$ can be identified with work $\bigO((\log n)^2)$, as $z$ is one of at most $\log n$ many nodes on the path from $x$ to the leaf of $p$, and one needs to check that all nodes between $x$ and $z$ are effect-preserving.
As the auxiliary relations need to be updated for $\log n$ many nodes, the overall work of $\prog$ is $\bigO((\log n)^3)$. We refer to the appendix for more details.
\end{proofsketch}

\subsubsection*{A work-efficient dynamic program for range queries for $D_1$ and $D_k$}
Unfortunately, the program of Theorem~\ref{thm:D1polylog} does not support range queries, since it seems that 
one would need to combine the unmatched parentheses of $\log n$ many nodes of the binary tree in the worst case.
However, its idea can be combined with the idea of Proposition \ref{prop:monoidWorkBound}, yielding a program that maintains $\ell$ and $r$  for $\bigO(n^{\epsilon})$ special intervals on a constant number of levels.

In fact, this approach even works for $D_k$ for $k>1$. Indeed, with the help of $\ell$ and $r$, it is possible to identify for each position of an opening parenthesis the position of the corresponding closing parenthesis in $\bigO(1)$ parallel time with work~$n^\epsilon$, and then one only needs to check that they match everywhere. The latter contributes an extra factor $\bigO(n)$ to the work, for $k>1$, but can be skipped for~\mbox{$k=1$}.  

\begin{theorem}\label{theo:rangeDkWorkBound}
    For all $\epsilon > 0$, $k > 1$, 
    \begin{enumerate}[a)]
        \item $\RangeMember{D_1}$ can be maintained in \DynFO with $\bigO(n^{\epsilon})$ work, and
        \item $\RangeMember{D_k}$ can be maintained in \DynFO with $\bigO(n^{\epsilon})$ work per change operation and $\bigO(n^{1+\epsilon})$ work per query operation.
        \end{enumerate}
\end{theorem}
\begin{proofsketch}
    In the following we reuse the definition of \emph{special intervals} from the proof of Proposition \ref{prop:monoidWorkBound}
    as well as the definition of $\ell$ and $r$ from the proof of Proposition \ref{thm:D1polylog}.
    We first describe a dynamic program $\prog$ for $\RangeMember{D_1}$.
   It maintains $\ell$ and $r$ for all special intervals, which is clearly doable with  $\bigO(n^{\epsilon})$ work per change operation. 
    Similar to the proof of Proposition \ref{prop:monoidWorkBound}, the two crucial observations (justified in the appendix) are that
    (1) a range query can be answered with the help of a constant number of special intervals, and
    (2) the change operation affects only a bounded number of special intervals per level. 

    \newcommand{\diff}{\textnormal{d}}


    As stated before, the program for $\RangeMember{D_k}$ also maintains $\ell$ and~$r$, but it should be emphasised that also in the case of several parenthesis types, the definition of these functions ignores the bracket type. With that information it computes, for each opening bracket the position of its matching closing bracket, with the help of $\ell$ and $r$, and checks that they match. This can be done in parallel and with work $\bigO(n^\epsilon)$ per position. We refer to the appendix for more details.
\end{proofsketch}

\subsubsection*{Moderately work-efficient dynamic programs for $D_k$}

We now turn to the membership query for $D_k$ with $k>1$.
Again, our program basically mimics the sequential algorithm from \cite{FrandsenHusfeldt+1995} which heavily depends on the dynamic problem $\StringEquality$ that asks whether two given strings are equal.

\dynProblemDefinition{$\StringEquality$}
    {Two Sequences $u = u_0 \ldots u_{n-1}$ and $v = v_0 \ldots v_{n-1}$ of letters with $u_i, v_i \in \Sigma \cup \{\epsilon\}$}
    {\begin{qclist}
        \qclitem{$\set_{x,\sigma}(i)$ for $\sigma \in \Sigma, x \in \{u,v\}$:}{}{Sets $x_i$ to $\sigma$, if $x_i = \epsilon$}
        \qclitem{$\reset_x(i)$ for $x \in \{u,v\}$:}{}{Sets $x_i$ to $\epsilon$}
    \end{qclist}}
    {\begin{qclist}
        \qclitem{$\qequals$:}{}{Is $u_0 \circ \ldots \circ u_{n-1} = v_0 \circ \ldots \circ v_{n-1}$?}
    \end{qclist}}

It is easy to show that a linear amount of work is sufficient to maintain $\StringEquality$.
\begin{lemma}\label{lem:stringEqualityLinear}
    $\StringEquality$ is in $\DynFO$ with work $\bigO(n)$.
\end{lemma}

Because of the linear work bound for $\StringEquality$ our dynamic program for $\Member{D_k}$ also has a linear factor in the work bound.
\begin{theorem}\label{lem:Dklinearithmic}
    $\Member{D_k}$ is maintainable in $\DynFO$ with $\bigO(n \log n  + (\log n)^3)$ work for every fixed $k \in \N$.
\end{theorem}
\begin{proofsketch}
    The program can be seen as an extension of the one for $\Member{D_1}$.
    As unmatched parentheses are no longer well-defined if we have more than one type of parenthesis
    the idea of \cite{FrandsenHusfeldt+1995} is to maintain the parentheses to the left and right that remain
    if we reduce the string by matching opening and closing parentheses regardless of their type.
    To be able to answer $\Member{D_k}$, the dynamic program maintains the unmatched parentheses for every node $x$ of a tree spanning the input word,
    and a bit $M(x)$ that indicates whether the types of the parentheses match properly.

    How the unmatched parentheses can be maintained for a node $x$ after a change operation depends on the ``segment'' of $\str(x)$ in which the change happened and
    in some cases reduces to finding a node $z$ with a local property on the path from $x$ to the leaf that corresponds to the changed position.

    To update $M(x)$ for a node $x$ with children $y_1$ and $y_2$ the dynamic program compares the unmatched parentheses to the right of $y_1$
    with the ones to the left of $y_2$ using $\StringEquality$.
    We refer to the appendix for more details.
\end{proofsketch}

\newcommand{\leftS}[1]{\ensuremath{\redu^{\ell}(#1)}}
\newcommand{\rightS}[1]{\ensuremath{\redu^{r}(#1)}}
\newcommand{\reda}{\mu_a}
\newcommand{\redu}{\mu_u}

%

Maintaining string equality and membership in $D_k$ for $k>1$ is even closer related which is stated in the following lemma.

\begin{lemma}\label{lem:DkvsStringEquality}
    \begin{enumerate}[a)]
        \item If $\StringEquality$ can be maintained in $\DynFO$ with work $W(n)$ then $\Member{D_k}$ can be maintained in $\DynFO$ with work $\bigO(W(n) \cdot \log n + (\log n)^3)$, for each $k \geq 1$.
        \item If $\Member{D_k}$ can be maintained in $\DynFO$ with work $W(n)$ for all $k$, then $\StringEquality$ can be maintained in $\DynFO$ with work $\bigO(W(n))$.
    \end{enumerate}
\end{lemma}


    \section{Conclusion}\label{section:conclusion}
    In this paper we proposed a framework for studying the aspect of work for the dynamic, parallel complexity class \DynFO. We established that all regular languages can be maintained in $\DynFO$ with $\bigO(n^\epsilon)$ work for all $\epsilon > 0$, and even with $\bigO(\log n)$ work for star-free regular languages. For context-free languages we argued that it will be hard to achieve work bounds lower than $\bigO(n^{\omega - 1 - \epsilon})$ in general, where $\omega$ is the matrix multiplication exponent. For the special case of Dyck languages $D_k$ we showed that $\bigO(n \cdot (\log n)^3)$ work suffices, which can be further reduced to $\bigO(\log^3 n)$ work for $D_1$. For range queries, dynamic programs with work $\bigO(n^{1+\epsilon})$ and $\bigO(n^\epsilon)$ exist, respectively.

We highlight some research directions. One direction is to improve the upper bounds on work obtained here. For instance, it would be interesting to know whether all regular languages can be maintained with polylog or even $\bigO(\log n)$ work and how close the lower bounds for context-free languages can be matched. Finding important subclasses of context-free languages for which polylogarithmic work suffices is another interesting question.
Apart from string problems, many \DynFO results concern problems on dynamic graphs, especially the reachability query \cite{DattaKMSZ18}. How large is the work of the proposed dynamic programs, and are more work-efficient dynamic programs possible?

The latter question also leads to another research direction: to establish further lower bounds. The lower bounds obtained here are relative to strong conjectures. Absolute lower bounds are an interesting goal which seems in closer reach than lower bounds for \DynFO without bounds on the work.


    \bibliography{bibliography}
    \appendix
   \newpage
      \section{Appendix for Section~\ref{section:model}}
   \begin{proofsketch}[for Proposition~\ref{prop:PseudoDynFOvsPRAM}]
  We briefly sketch, how  a PRAM $\prog$ can evaluate an update rule $\pi$ of the form
   \begin{algorithmic}[0]
    \OnChange{$\delta(\tpl p)$}
        \UpdateAtWhereBy{$R$}{$(t_1(\tpl p, \tpl x), \ldots, t_k(\tpl p, \tpl x))$}{$x_1 \le g_1(n) \land \ldots \land x_\ell \le g_\ell(n)$}
            $P$.
        \EndUpdateAtWhereBy
    \EndOnChange
  \end{algorithmic}
  with $P=P_1;P_2$ as before, 
  in constant time and with work $\bigO(\workbound(\pi))$. The initial procedure $P_1$ has to be evaluated only once and therefore contributes at most $\workbound(P_1)$. 
   The main procedure $P_2$ has to be evaluated once, for every tuple $\tpl x $ that fulfils $x_1 \le g_1(n) \land \ldots \land x_k \le g_\ell(n)$. To this end, it can use $m \df  g_1(n) \cdot \ldots \cdot g_\ell(n)$ processors, each taking care of one evaluation, as in the following pseudo-code.
   \begin{algorithm}
     \begin{algorithmic}[1]
       \Procedure{$\text{Evaluate}_{\delta}^{R}$}{$\tpl p$}
               \State Evaluate $P_1$
               \Parfor{$x_1 = 1\ \To\ g_1(n)$}
                   \State $\ddots$
                   \Parfor {$x_\ell = 1\ \To\ g_\ell(n)$}
                       \State $R\big(t_1(\tpl p, \tpl x), \ldots, t_k(\tpl p, \tpl x)\big) \gets \Call{$\text{Eval}_{P_2}$}{\tpl p, \tpl x}$
                   \EndFor
               \EndFor
           \EndProcedure
       \end{algorithmic}
       \caption{Evaluation of the update rule for $\delta(\tpl p)$.}
       \label{alg:EvalUpdateRule}
     \end{algorithm}
       This amounts to
       at most  $m$ evaluations of $P_2$ and  it thus suffices to show that one evaluation of $P_2$ only requires work $\bigO(\workbound(P_2))$.

       The proof of this is relatively straightforward, but tedious. Each computation path consists only of a constant number of steps (unless a query towards a supplementary instance occurs). Thus, assuming sufficient parallelisation, the PRAM only needs time $\bigO(1)$.  Along each path, $\workbound$ basically yields only 1, but any constant number is captured by the big $\bigO$. This explains, why we can use $\max$ instead of addition, everywhere. 
       Parallel branches, \textbf{exists}-expressions and \textbf{unique}-expressions translate into parallel branching of the PRAM, and that is captured in the definition of $\workbound$ by respective factors. Likewise, query operations to supplementary instances are captured, by respective factors in the definition of $\workbound$.
\end{proofsketch}


  \newpage
      \section{Appendix for Section~\ref{section:regular}}
   
\begin{proofsketch}[of Proposition~\ref{prop:monoidWorkBound}, continued]
    We start by introducing necessary notation.
    In the following we fix the length $n$ of a sequence $m_0,\ldots,m_{n-1}$ of monoid elements, we chose $t$ as described above, and assume that $t^h=n$.
    We write $m[\ell,r)$ for the subsequence $m_\ell,\ldots,m_{r-1}$, and, depending on the context, for the product $m_\ell\circ\ldots\circ m_{r-1}$. If $r \le \ell$ then $m[\ell,r)$ is the neutral monoid element.

    In order to define special intervals, we consider positions $j\in[n]$ in $t$-adic representation as \( \sum_{i=0}^{h-1} j_i t^i \),
    where $j_i\in[t]$, for every $0 \le i < h$. 
    We say that a number $j$ has \emph{height} $k$ if $k < h$ is the maximum number such that \mbox{\( \sum_{i=k}^{h-1} j_i t^i= \sum_{i=0}^{h-1} j_i t^i\)}, that is, if $j_0,\ldots,j_{k-1}$ are zero and $j_k$ is not (or $j=0$ and $k=h-1$). 

    We say that an interval $[i,j)$ of positions is \emph{special} if, for
    some $k$ and $p$, the height of $i$ and $j$ is at
    least $k$ and  $pt^{k+1}\le i \le j \le (p+1)t^{k+1}$.  
    It is easy to see that, for each $k$ and $p$ there are $\bigO(t^2)$ such special intervals.

    The dynamic program maintains the product $m[i,j)$ for all special intervals $[i,j)$.

    To show that the two crucial observations (1) and (2) stated above are correct, we need some more notation.
    For a number $j\in[n]$ we write $\tpred{j}{k}$ for the largest number $p$ of height at least $k$ with $p\le j$ and $\tsucc{j}{k}$ for the smallest number $q$ of height at least $k$ with $j\le q$,
    i.e. 
     
    \[\tpred{j}{k}\df\displaystyle j_kt^k + \sum_{i=k+1}^{h-1} j_i t^i \quad \text{ and } \quad \tsucc{j}{k}\df\displaystyle (j_k+1) t^k+\sum_{i=k+1}^{h-1} j_i t^i.\]
%
    Towards observation (1), it is not hard to see that the following equality holds for all $0 \leq \ell, r \leq n-1$ and that the product consists of $2(h-1)+1=\bigO(1)$ factors. Suppose that $k < h$ is the smallest number such that $\tsucc{\ell}{k+1} > \tpred{r}{k+1}$, then:
    \[
        m[\ell,r)= \prod_{j=0}^{k-1} m[\tsucc{\ell}{j},
        \tsucc{\ell}{j+1}) \; \circ \;
            m[\tsucc{\ell}{k},\tpred{r}{k}) \; \circ \; 
            \prod_{j=0}^{k-1} m[\tpred{r}{j+1}, \tpred{r}{j})
    \]
    It is thus straightforward to answer range queries by a dynamic program with work $\bigO(1)$.

    Furthermore, towards observation (2), if position $p$ is changed, only the information for the $ht^2$ special intervals that include $p$ needs to be updated.
    This can be done in $h=\bigO(1)$ phases. In phase $k$, the special intervals with height $k$ are updated.
    The new product $m'[i,j)$ can be obtained via
			\[m'[i,j) = m[i,\tpred{p}{k}) \circ m'[\tpred{p}{k},\tsucc{p}{k}) \circ m[\tsucc{p}{k},j),\]
    where $m'[\tpred{p}{k},\tsucc{p}{k})$ is the already computed new value of the special interval $[\tpred{p}{k},\tsucc{p}{k})$ of height $k-1$. 
    This completes the proof sketch for Proposition~\ref{prop:monoidWorkBound}.
\end{proofsketch}

Towards a proof of Theorem~\ref{theo:FO-work}, 
we employ the following algebraic property of syntactic monoids of
star-free languages, due to Schützenberger, McNaughton and Papert
(see, e.g., \cite[Chapter V.3]{Straubing1994}). A monoid $M$ is called
a \emph{group} if for each $m \in M$ there is an inverse $m^{-1}$ with
$m \cdot m^{-1} = 1$. A monoid is called \emph{group-free} if it does
not contain a nontrivial group.  

\begin{lemma}[{\cite[Theorem V.3.2]{Straubing1994}}]\label{lem:FO-group-free}
    A language $L$ is star-free iff $M(L)$ is finite and group-free.
\end{lemma}

In \cite{FrandsenMiltersen+1997} it was shown that $\RangeEval{M}$ for a group-free monoid $M$ can be maintained in $\bigO(\log \log n)$ sequential time on RAMs with cell size $\bigO(\log n)$. We adapt their algorithm to show that $\RangeEval{M}$ can be maintained in \DynFO with logarithmic work. 

The main ingredient for the algorithm of  \cite{FrandsenMiltersen+1997} is the following consequence of a decomposition theorem~by Krohn and Rhodes.

\begin{lemma}[\cite{KrohnRhodes1965}\cite{FrandsenMiltersen+1997}]\label{lem:KrohnRhodes}
    Let $M$ be a finite, group-free monoid. Then one of the following holds:
    \begin{enumerate}[(a)]
        \item $M = \{1\}$,
        \item there is a $k$ such that $M - \{1\} = \{\sigma, \sigma^2, \ldots, \sigma^k = \sigma^{k+1}\}$,
        \item $\sigma\sigma' = \sigma$, for all $\sigma,\sigma' \in M - \{1\}$, or
        \item $M = V \cup T$ where $T \neq M$ and $V \neq M$ are submonoids of $M$ and
            $T - \{1\}$ is a left ideal of $M$, i.e. $\sigma_M \sigma_T \in T - \{1\}$ for all $\sigma_M \in M$ and $\sigma_T \in T - \{1\}$
    \end{enumerate}
\end{lemma}

We illustrate this lemma by the following example. It will be reused
to illustrate the construction in our dynamic program for maintaining $\RangeEval{M}$. 

\begin{example}\label{ex:FOdecomposition}
    Consider the language $L = L\big(c^\ast a c^\ast a c^\ast b (a+b+c)^\ast\big)$ that contains all words over the alphabet $\Sigma = \{a,b,c\}$
    having at least one $b$ and at least two $a$'s somewhere before the first $b$.
    It syntactic  monoid is $M = \{ 1,A,A^2,B,D,E \}$ with the
    following multiplication table:

    \begin{table}[H]
        \begin{center}
            \begin{tabular}{l|llllll}
                ${}_x\setminus {}^y$    & $1$   & $A$   & $A^2$ & $B$   & $D$   & $E$   \\ \hline 
                $1$                     & $1$   & $A$   & $A^2$ & $B$   & $D$   & $E$   \\
                $A$                     & $A$   & $A^2$ & $A^2$ & $D$   & $E$   & $E$   \\
                $A^2$                   & $A^2$ & $A^2$ & $A^2$ & $E$   & $E$   & $E$   \\
                $B$                     & $B$   & $B$   & $B$   & $B$   & $B$   & $B$   \\
                $D$                     & $D$   & $D$   & $D$   & $D$   & $D$   & $D$   \\
                $E$                     & $E$   & $E$   & $E$   & $E$   & $E$   & $E$
            \end{tabular}
        \end{center}
        \label{tbl:exampleFO}
      \end{table}
The elements of $M$ (and thus the states of the minimal automaton for
$L$) correspond to the following languages.  In particular, $E$ corresponds to $L$.
  \begin{table}[H]
        \begin{center}
            \begin{tabular}{c|l}
                 $1$    & $L\big(c^\ast\big)$\\
            $A$    & $L\big(c^\ast a c^\ast\big)$\\
            $A^2$  & $L\big(c^\ast a c^\ast a c^\ast\big)$\\
            $B$    & $L\big(c^\ast b (a+b+c)^\ast\big)$ \\
            $D$    & $L\big(c^\ast a c^\ast b (a+b+c)^\ast\big)$\\
            $E$    & $L\big(c^\ast a c^\ast a c^\ast b (a+b+c)^\ast\big)$
            \end{tabular}
        \end{center}
        \label{tbl:monoid-languages}
      \end{table}
     

    The monoid can be decomposed as $M=V\cup T$ with $T = \{ 1,B,D,E \}$ and $V = \{ 1,A,A^2 \}$,
    which are Type (c) and Type (b) submonoids, respectively, of $M$
    in the sense of Lemma \ref{lem:KrohnRhodes}.
\end{example}


The characterization of group-free monoids provided by Lemma \ref{lem:KrohnRhodes} is the basis for a recursive approach for maintaining $\RangeEval{M}$.

\newcommand{\res}{\textnormal{res}}
    \newcommand{\rangev}{\textnormal{range}_v}
    \newcommand{\ranget}{\textnormal{range}_t}
    \newcommand{\rangeu}{\textnormal{range}_u}
    \newcommand{\T}{T_{\neq 1}}
\begin{theorem}\label{thm:group-free-work}
    Let $M$ be a group-free monoid. Then $\RangeEval{M}$ can be maintained in \DynFO with work $\bigO(\log n)$
    for each change and query operation.
\end{theorem}
\begin{proof}   
    The proof is by induction on the decomposition of $M$. We describe, for each case of Lemma~\ref{lem:KrohnRhodes},
    how a dynamic program $\prog$ can maintain $\RangeEval{M}$ with $\bigO(\log n)$ work.
    
    Let $m = m_0 \ldots m_{n-1}$ be an input sequence. Case (a) from Lemma ~\ref{lem:KrohnRhodes} is trivial.

    For Case (b), suppose that $M = \{1, \sigma, \sigma^2, \ldots, \sigma^k = \sigma^{k+1}\}$ for some fixed $k$. As in this case the monoid elements (except for $1$) only differ in their exponents, the result of a range query $\qrange(\ell,r)$ can be computed by summing over the exponents of all elements $\neq 1$ between $\ell$ and $r$. A closer look reveals that it suffices to sum over the exponents of the first up to $k$ elements differing from the identity element $1$, because $\sigma^k = \sigma^{k+1}$. For finding these up to $k$  elements quickly, the program uses a supplementary instance of \NextInK for the set $K = \{ i \mid w_i \neq 1 \}$. 
    
    These ideas lead to Algorithm \ref{alg:RangeUpdatesB} for answering range queries. As only a constant number of successors are queried from $K$, it requires at most $\bigO(\log n)$ work due to Lemma \ref{lem:NextInKWork}. For changes of the input instance, only the supplementary \NextInK instance needs to be updated, which also requires at most $\bigO(\log n)$ work by Lemma \ref{lem:NextInKWork}.

    \newcommand{\sumVar}{s}
    \newcommand{\curVar}{i}
    \newcommand{\exponent}{\textnormal{exponent}}
    \begin{algorithm}[t!]
        \begin{algorithmic}[1]
            \OnQuery{$\qrange(\ell, r)$}
                \State $\sumVar \gets 0$
                \State $\curVar \gets (\ell-1)$
                \For{$1\ \To\ k$}
                    \State $\curVar = K.\mtext{succ}(\curVar)$
                    \If{$\curVar \le r$}
                        \State $\sumVar \gets \sumVar+\exponent(\curVar)$
                    \EndIf
                \EndFor
                \If{$s>0$}
                    \State \Return $a^{\sumVar}$
                \Else
                    \State \Return $1$
                \EndIf
            \EndOnQuery
        \end{algorithmic}
        \caption{Querying the product $m_{\ell} \cdot \ldots \cdot m_r$ in Case (b) of the proof of Theorem \ref{thm:group-free-work}.}
        \label{alg:RangeUpdatesB}
    \end{algorithm}

    For Case (c), suppose that $M$ is a monoid with $\sigma\sigma' = \sigma$ for all $\sigma,\sigma' \in M - \{1\}$. In this case, a range query $\qrange(\ell,r)$ results in the first element in the interval $[\ell,r]$ which is not the identity element (if such an element does not exist, then the result is the identity element). In order to find such an element quickly, the dynamic program maintains a supplementary \NextInK instance for the set $K = \{ i \mid w_i \neq 1 \}$.
    
    Now range queries can be answered according to Algorithm \ref{alg:RangeUpdatesC}, which clearly requires at most $O(\log n)$ work. For changes of the input instance, again only the supplementary \NextInK instance needs to be updated, which also requires at most $\bigO(\log n)$ work.

    \begin{algorithm}[t!]
        \begin{algorithmic}[1]
            \OnQuery{$\qrange(\ell,r)$}
                    \State $i \gets K.\mtext{succ}(\ell-1)$
                    \If{$i \le r$}
                        \State \Return $m_i$
                    \Else
                        \State \Return $1$
                    \EndIf
            \EndOnQuery
        \end{algorithmic}
        \caption{Querying the product $m_{\ell} \cdot \ldots \cdot m_r$ in Case (c)  of the proof of Theorem \ref{thm:group-free-work}.}
        \label{alg:RangeUpdatesC}
    \end{algorithm}
    

    For Case (d), suppose that $M = V \cup T$ where $T \neq M$ and $V \neq M$ are submonoids of $M$ and $T - \{1\}$ is a left ideal of $M$.
    Because of $T \neq M$ and $V \neq M$, by induction we can assume that $\RangeEval{T}$ and $\RangeEval{V}$ can be maintained with work $\bigO(\log n)$. 
    
    Denote $T - \{1\}$ by $\T$. The idea of the dynamic program is to split the input sequence~$m$ into $\T$-blocks, consisting of a maximal sequence of consecutive elements of $\T$ only, and $\overline{\T}$-blocks consisting of maximal sequence of elements of $\overline{\T}$ only. The program maintains precomputed  products of each $\overline{\T}$-block combined with its successive $\T$-block. Because $\T$ is a left ideal, all these precomputed block products result in elements of $\T$. A range query $\qrange(\ell,r)$ for $m$ can then be answered by computing the product of these precomputed block products between $\ell$ and $r$ and adding the at most two incomplete block products at the beginning and the end of the queried range.

    For encoding the blocks as well as the precomputed products, the program uses three sequences $t$, $v$ and $u$ with the following intention. The sequences $v$ and $t$ partition $m$ into elements of $\T$ and $V-\T$, respectively. For the definition of $u$, let $K = \{ i \mid w_i \in \T \land w_{i+1} \notin \T\}$ be the set of \emph{switching positions}, i.e.\ positions of the input sequence where the monoid element changes from $\T$ to $V-\T$. At a switching position $i$, $u_i$ stores the product of elements between the previous switching position and $i$. In order to find switching positions quickly, the dynamic program maintains a supplementary \NextInK instance for the set~$K$.
%
%
%
    
    More precisely, the three sequences are defined as
    \begin{align*}
        t_i = &
            \begin{cases}
                m_i    & \textnormal{if $m_i \in \T$}\\
                1      & \textnormal{otherwise}
            \end{cases}, \\
        v_i = &
            \begin{cases}
                m_i    & \textnormal{if $m_i \in V-\T$}\\
                1      & \textnormal{otherwise}
            \end{cases},       \\
        u_i = &
            \begin{cases}
                \prod_{j=k_i+1}^{i}m_j   &\textnormal{if $i \in K$}\\
                1_T                    & \textnormal{otherwise}
            \end{cases},
    \end{align*}
    where $k_i = K.\mtext{pred}(i)$ if $K.\mtext{pred}(i) \neq \bot$ and $k_i=0$ otherwise.
    Note that $m_i \in \T$ for all $i$ with $u_i \neq 1$ because $\T$ is a left interval. We refer to Example \ref{example:FOcaseD} for an illustration of these sequences. 

    The result of a range query $\qrange(\ell,r)$ can then be computed as
    \begin{equation*}
        \prod_{i=\ell}^{r}m_i = \prod_{i=\ell}^{k_1}v_i \prod_{i=\ell}^{k_1}t_i \prod_{i=k_1+1}^{k_q}u_i \prod_{i=k_q+1}^{r}v_i \prod_{i=k_q+1}^{r}t_i
    \end{equation*}
    where $k_1 < \ldots < k_q$ are all positions in $K \cap \{ \ell-1, \ldots, r\}$. The resulting dynamic program, see Algorithm \ref{alg:RangeUpdatesD}, clearly requires at most $O(\log n)$ work.

    \begin{algorithm}[t!]
        \begin{algorithmic}[1]
            \OnQuery{$\qrange(\ell,r)$}
                    \State $k_1 \gets K.\mtext{succ}(\ell-1)$
                    \State $k_q \gets K.\mtext{pred}(r)$
                    \State \Return $v.\qrange(\ell,k_1) \circ t.\qrange(\ell,k_1) \circ u.\qrange(k_1+1,k_q)$ \\ \hspace{5.0cm}$\circ \; v.\qrange(k_q+1,r) \circ t.\qrange(k_q+1,r)$
            \EndOnQuery
        \end{algorithmic}
        \caption{Querying the product $m_{\ell} \cdot \ldots \cdot m_r$ in Case (d)  of the proof of Theorem \ref{thm:group-free-work}.}
        \label{alg:RangeUpdatesD}
    \end{algorithm}
    

    For changes of the input instance, the supplementary \NextInK instance for the switching positions needs to be updated, which requires at most $\bigO(\log n)$ work.
    Also the sequences $t$, $v$ and $u$ need to be updated. For $v$ and $t$ this is easy, but updating $u$ requires some effort.
    
    For maintaining $u$ after a change of the input instance at position $p$, we distinguish the cases as summarized in Table \ref{tbl:FO-cases-update-u}.
    In every case, the precomputed block product stored in the next switching position after $p$ needs to be recomputed.
    Additionally, if the insertion created a new block by splitting a block (see Cases (4) and (5)) or removed a block by merging two blocks (see Cases (2) and (3)) either $u_{p}$ or $u_{p-1}$ has to be updated as well.

    \begin{table}
        \begin{center}
            \scalebox{0.85}{
            \begin{tabular}{l|l}
                Case condition & Changes to $u$ \\ \hline

                \begin{minipage}{0.5\textwidth}
                    \begin{enumerate}[(1)]
                        \item Switching positions do not change, i.e. either
                            \begin{enumerate}[(a)]
                                \item $m_p \in \T$ iff $m'_p \in \T$,
                                \item $m_{p-1} \notin \T$ and $m'_p, m_{p+1} \in \T$, or 
                                \item $m_{p-1}, m'_p \notin \T$ and $m_{p+1} \in \T$
                            \end{enumerate}
                    \end{enumerate}
                \end{minipage}
                & $\displaystyle u_k \gets \prod_{i=j+1}^k v_i \circ \prod_{i=j+1}^k t_i$\\ \hline

                \begin{minipage}{0.5\textwidth}
                    \begin{enumerate}[(2)]
                        \item Two $\T$-blocks are merged, i.e.
                            \begin{enumerate}[]
                                \item $m_{p-1}, m'_p, m_{p+1} \in \T$ and $m_p \notin \T$,
                            \end{enumerate}
                    \end{enumerate}
                \end{minipage}
                & $u_k$ as in Case (1) and $u_{p-1} \gets 1$\\ \hline

                \begin{minipage}{0.5\textwidth}
                    \begin{enumerate}[(3)]
                        \item Two $\overline{\T}$-blocks are merged, i.e.
                            \begin{enumerate}[]
                                \item $m_{p-1}, m'_p, m_{p+1} \notin \T$ and $m_p \in \T$,
                            \end{enumerate}
                    \end{enumerate}
                \end{minipage}
                & analogous to (2)\\ \hline

                \begin{minipage}{0.5\textwidth}
                    \begin{enumerate}[(4)]
                        \item A $\T$-block is split, i.e.
                            \begin{enumerate}[]
                                \item $m_{p-1}, m_p \in \T$ and $m'_p \notin \T$,
                            \end{enumerate}
                    \end{enumerate}
                \end{minipage}
                & $\displaystyle u_k \gets \prod_{i=j+1}^{p-1} v_i \circ \prod_{i=j+1}^{p-1} t_i$ and $\displaystyle u_{p-1} \gets \prod_{i=p}^{k} v_i \circ \prod_{i=p}^{k} t_i$, \\ \hline

                \begin{minipage}{0.5\textwidth}
                    \begin{enumerate}[(5)]
                        \item A $\overline{\T}$-block is split, i.e.
                            \begin{enumerate}[]
                                \item $m'_p, m_{p+1} \notin \T$ and $m_p \in \T$,
                            \end{enumerate}
                    \end{enumerate}
                \end{minipage}
                & analogous to (4)

            \end{tabular}
        }
        \end{center}
        \caption{
            Summary of the cases for updating $u$ after a change at position $p$.
            Suppose that $v, t$ and $K$ are already updated.
            By $m_p$ and $m'_p$ we denote the elements at position $p$ before and after the change operation, respectively.
            In all cases $k = K.\mtext{succ}(p)$ is defined as the next switching position after~$p$.
            For Case (1), $j$ is defined as $K.\mtext{pred}(p)$; for Case (3) as $K.\mtext{pred}(p-1)$.
        }
        \label{tbl:FO-cases-update-u}
    \end{table}

    In all cases, $u$ is updated with a constant number of range queries to $v$ and $t$ which results overall in $\bigO(\log n)$ work.
\end{proof}

Combining Lemma~\ref{lem:FO-group-free} and Theorem~\ref{thm:group-free-work} we get the desired upper work bound for \FO definable languages.

    \begin{example}\label{example:FOcaseD}
        We revisit the monoid $M$ with decomposition into $T$ and $V$ introduced in Example \ref{ex:FOdecomposition}. Consider the sequence $m$ with its sequences $t,v$ and $u$ as defined in the proof of Theorem \ref{thm:group-free-work}:
        \definecolor{cadmiumgreen}{rgb}{0.0, 0.42, 0.24}
        \newcommand{\highlightU}[1]{\underline{#1}}
        \newcommand{\highlightR}[1]{{\color{red}#1}}
        \newcommand{\highlightG}[1]{{\color{cadmiumgreen}#1}}
        \newcommand{\highlightB}[1]{{\color{blue}#1}}
        \begin{align*}
                &\texttt{\ \ \ \ \ \ \ \ \ \ 1\ \ \ \ \ \ \ \ \ 2\ \ \ \ \ \ \ \ \ 3}\\
                &\texttt{01234567890123456789012345678901}\\
            m = &\texttt{A\highlightU{A1AABBA11A1AABAA1B11A1B1A1}A11BB}\\
            t = &\texttt{1\highlightB{1111BB}1111111B111B1111B\highlightB{111}111BB}\\
            v = &\texttt{A\highlightR{A1AA11}A11A1AA1AA1111A11\highlightR{1A1}A1111}\\
            u = &\texttt{111111E\highlightG{1111111E111E1111D}1111111E}
        \end{align*}
        The sequences $t$ and $u$ consist of $\T$-elements and $v$ of $V-\T$-elements only.
        For a query $\qrange(1,26)$ on $m$ the set of switching positions is $\{6,14,18,23\}$, so the result of the query can be computed as follows:
        \begin{align*}
            \prod_{i=1}^{26}m_i &= \highlightR{\prod_{i=1}^{6}v_i} \highlightB{\prod_{i=1}^{6}t_i} \highlightG{\prod_{i=7}^{23}u_i} \highlightR{\prod_{i=24}^{26}v_i} \highlightB{\prod_{i=24}^{26}t_i}\\
                                &= \highlightR{\texttt{A}^2} \highlightB{\texttt{B}} \highlightG{\texttt{E}} \highlightR{\texttt{A}} \highlightB{\texttt{1}}\\
                                &= \texttt{E}
        \end{align*}
        Note, that the results of the five subproducts are determined by subqueries to $t,v$ and $u$ and can be evaluated as described in Case (b) for $v$ and Case (c) for $t$ and $u$.
    \end{example}


  \newpage
      \section{Appendix for Section~\ref{section:cfl}}

\begin{proofsketch}[of Lemma ~\ref{lem:staticUsingDyn}]
    Let $\prog$ be a dynamic program that maintains $\Member{L}$ with work $f(n)$.
    Whether $w \in L$ holds for a given input string $w = \sigma_0 \sigma_1 \ldots \sigma_{n-1}$ can be decided
    by simulating $\prog$ on the sequence of change operations that inserts $\sigma_1$ to $\sigma_n$ one after the other.
    By simulating the PRAM sequentially, the work $\bigO(f(n))$ for one change step translates into  sequential time $\bigO(f(n))$, thus yielding an overall sequential time  $\bigO(n \cdot f(n))$.
\end{proofsketch}

\begin{proof}[Proof of Theorem~\ref{theo:cfg-clique}]
  Let $G$ be the grammar from  Theorem~\ref{theo:cfg-clique} and let us assume, towards a contradiction, that there is  a dynamic program for $L(G)$ with work $\bigO(n^{\omega-1-\epsilon})$. 
  By Lemma~\ref{lem:staticUsingDyn} this would yield  an algorithm for the word problem for $L(G)$ with time $\bigO(n^{\omega-\epsilon})$. Theorem~\ref{theo:cfg-clique} would then give an algorithm for $k$-Clique with time bound $\bigO(n^{(\frac{k}{3}+1)(\omega-\epsilon)})$. 

  However, for $k\ge \frac{18}{\epsilon}$, this yields a contradiction as follows.
  
  We have $(\frac{k}{3}+1)(\omega-\epsilon)=\frac{\omega}{3}k+\omega-(\frac{\epsilon}{3}k +\epsilon)\le \frac{\omega}{3}k-\frac{\epsilon}{6}k$, since $\frac{\epsilon}{3}k +\epsilon=\frac{\epsilon}{6}k + \frac{\epsilon}{6}k +\epsilon\ge 3+\frac{\epsilon}{6}k +\epsilon$.
  Thus, $k$-Clique could be solved in time $\bigO(n^{(1-\frac{\epsilon}{6})k})$.
\end{proof}

\begin{proof}[Proof of Theorem~\ref{thm:D1polylog}]
  We give a complete proof here, repeating parts of the proof sketch in the body, for convenience.
  
With a string $w$ over $\Sigma_1$ we associate the \emph{reduced string} $\mu(w)$ that basically consists of the unmatched closing parentheses of $w$ followed by the unmatched opening parentheses of $w$. 

Formally, $\red(w)$ can be defined for strings $w \in \Sigma_1^\ast$ as $\red(\epsilon) = \epsilon$ and
\[
    \red(z \sigma) = 
        \begin{cases}
            z'              & \textnormal{if $\red(z)=z'\open$ and $\sigma = \close$}\\
            \red(z)\sigma & \textnormal{else}
        \end{cases}
\]
for every $\sigma \in \Sigma_1$ and $z \in \Sigma_1^\ast$.
As indicated before, a reduced string can be split into two strings $u$ and $v$ of the forms $\close^\ast$ and $\open^\ast$, respectively.
We call $u$ the unmatched (closing) parentheses of $w$ to the left and $v$ the unmatched (opening) parentheses of $w$ to the right.

The algorithm of \cite{FrandsenHusfeldt+1995} stores, for each node $x$ of $T$, the reduced string $\mu(\str(x))$ of its associated string, and updates these strings bottom up after each change, starting from the changed position, resulting in logarithmic (sequential)  update-time.
We show that the new reduced strings can be computed in constant parallel time. To this end,  we show that for each affected node $x$ of $T$ its new reduced string can be determined by a parallel computation that does not need to know the new reduced string of its affected child.

    We describe a dynamic program $\prog$ that maintains $\Member{D_1}$, using a binary tree $T$ as sketched above,  which is constructed by the first-order initialisation as in the proof of Lemma \ref{lem:NextInKWork}.
    For each node $x$ of the tree, $\prog$ represents its reduced string by  the number of unmatched closing parentheses of $\str(x)$ to the left and to the right as auxiliary functions $\ell(x)$ and $r(x)$, respectively.
   In particular, the current word is well-balanced if $r(\textnormal{root}) = \ell(\textnormal{root}) = 0$.

    In the following, we describe how $\prog$ can update $\ell(x)$ and $r(x)$ for a node $x$ after a closing parenthesis $\close$ was inserted at some position $p$.   Maintaining $\ell$ and $r$ for the other change operations is analogous.
    Note that $\ell$ and $r$ only need to be updated for nodes whose substring contains the changed position $p$, that is,  for  ancestors of the leaf at position~$p$.
    
    There are two types of effects that an insertion of a closing parenthesis could have on $x$:
    Either $\ell(x)$ is increased by one and $r(x)$ remains unchanged, or $r(x)$ is decreased by one and $\ell(x)$ remains unchanged.
    We denote these effects by the pairs $(+1,0)$ and $(0,-1)$, respectively.
    We next analyse how the  effect that applies to $x$ depends on the effect of its affected child. We will see then that the effect on the child needs not to be known to update $x$, but that it can be determined by a direct parallel exploration of the nodes on the path from $x$ to~$p$.

    To this end, let $y_1$ be the left child of $x$ and $y_2$ its right child. The effect on $x$ depends on whether the changed position $p$ is in $\str(y_1)$ or in $\str(y_2)$, and on the relation between $r(y_1)$ and $\ell(y_2)$ before the change.
   Table~\ref{tbl:table-d1-cases-app} lists all possible combinations for (1) whether the changed position $p$ is in $\str(y_1)$ or in $\str(y_2)$, (2) which effect occurs there, and (3) the relation between $r(y_1)$ and $\ell(y_2)$, and which effect is applied to $x$ in each case.
For example, if $p$ is in the subtree of $y_1$, the effect on $y_1$ is $(0,-1)$, and $r(y_1) \leq \ell(y_2)$ holds. Then, the inserted closing parenthesis matches (and is to the right of) a former unmatched opening parenthesis in $\str(y_1)$. As $r(y_1) \leq \ell(y_2)$, all unmatched opening parentheses of  $\str(y_1)$ could previously be matched   by unmatched closing parentheses  of $\str(y_2)$. Since after the change there is one unmatched opening parenthesis less in $\str(y_1)$, the unmatched opening parentheses of $y_1$ can still be matched but, on the other hand, one more closing parenthesis of $y_2$ cannot be matched in $\str(x)$, so the effect on $x$ is $(+1,0)$.

    \begin{table}
        \begin{center}
            \begin{tabular}{l|c|c}
                                        & $p$ is in $\str(y_1)$                                                    & $p$ is in $\str(y_2)$ \\ \hline
                $r(y_1) \leq \ell(y_2)$ & $\begin{array}{c}(+1,0) \to (+1,0) \\ (0,-1) \to (+1,0)\end{array}$   & $\begin{array}{c}(+1,0) \to (+1,0) \\ (0,-1) \to (0,-1)\end{array}$ \\ \hline
                $r(y_1) > \ell(y_2)$ & $\begin{array}{c}(+1,0) \to (+1,0) \\ (0,-1) \to (0,-1)\end{array}$   & $\begin{array}{c}(+1,0) \to (0,-1) \\ (0,-1) \to (0,-1)\end{array}$ \\ 
            \end{tabular}
        \end{center}
        \caption{
            The effect on $x$ after a closing parenthesis was inserted at position $p$.
            The effects depend on the effect on the children $y_1$ and $y_2$ of $x$:
            for example, an entry '$(0,-1) \to (+1,0)$' in the column '$p$ is in $\str(y_1)$' means that
            if the change operation has effect $(0,-1)$ on $y_1$ then the change operation has effect $(+1,0)$ on $x$.
        }
        \label{tbl:table-d1-cases-app}
    \end{table}
    A closer inspection of Table~\ref{tbl:table-d1-cases-app} reveals a crucial observation: in the  upper left and the lower right field of the table, the effect on $x$ is \emph{independent} of the effect on the child (being it $y_1$ or $y_2$). That is, these cases induce an effect on $x$ independent of the children. We thus call these cases \emph{effect-inducing}. In the other two fields, the effect on $x$ depends on the effect at the child, but in the simplest possible way: they are just the same. That is the effect at the child is just adopted  by $x$. We call these cases \emph{effect-preserving}. To determine the effect at $x$ it is thus sufficient to identify the highest affected descendant node~$z$ of~$x$, where an effect-inducing case applies, such that for all intermediate nodes between $x$ and $z$ only effect-preserving cases apply. 

    Algorithm~\ref{alg:D1} implements this idea. First it determines, for each ancestor $x$ of the change position $p$, whether it is effect-inducing (Ind) and which effect is induced (Ind$+$ for $(+1,0)$ and Ind$-$ for $(0,-1)$). Then it identifies, for each $x$, the node $z$ (represented by its height $i$ above $p$) as the unique effect-inducing node that has no effect-inducing node on its path to $x$.


    \begin{algorithm}[t]
        \begin{algorithmic}[1]
                    \OnChange{$\ins_\close(p)$}
			\With{}
			\Parfor{$1 \le i \le \log n$}
			    \State $x(i) \gets T.\anc(p,i)$ 
			    \State $c(i) \gets T.\anc(p,i-1)$ 
			 	\State $\text{Ind+}(i) \gets \fst(x(i)) = c(i) \land r(\fst(x(i))) \leq \ell(\snd(x(i)))$
			 	\State $\text{Ind-}(i) \gets \snd(x(i)) = c(i) \land r(\fst(x(i))) > \ell(\snd(x(i)))$	
			 	\State $\text{Ind}(i) \gets \text{Ind+}(i) \lor \text{Ind-}(i)$	
			\EndFor 	
			\EndWith            
            \UpdateAtWhereBy{$(\ell,r)$}{$T.\anc(p,k)$}{$k\le \log n$}
            \If{$k=0$}
                        \State \Return $(1,0)$
                    \Else\
                     \State $x \gets T.\anc(p,k)$
                     \If{\algexists{$1 \le i \le k$}{$\text{Ind}(i)$}}
                      \State $i \gets \algunique{1 \le i \le k}{\text{Ind}(i) \land \neg \algexists{i < j \le k}{\text{Ind}(j)}}$
                         \If{$\text{Ind+}(i)$}
                        \State \Return $(\ell(x) + 1,r(x))$
                    	\Else\
                    	\State \Return $(\ell(x),r(x)-1)$
                    \EndIf
                    	\Else\
                    	\State \Return $(\ell(x) + 1,r(x))$
                    \EndIf
                    \EndIf
            \EndUpdateAtWhereBy
            \EndOnChange   
        \end{algorithmic}
        \caption{Updating $(\ell,r)$ after the insertion of a closing parenthesis in the proof of Theorem \ref{thm:D1polylog}.}
        \label{alg:D1}
      \end{algorithm}

    The node $z$ can be identified with work $\bigO((\log n)^2)$, as $z$ is one of at most $\log n$ many nodes on the path from $x$ to $v$, and one needs to check that all nodes between $x$ and $z$ are effect-preserving.
    As the auxiliary relations need to be updated for $\log n$ many nodes, the overall work of $\prog$ is $\bigO((\log n)^3)$.

\end{proof}

\begin{proofsketch}[of Theorem~\ref{theo:rangeDkWorkBound}]
    In the following we reuse the definition of \emph{special intervals} from the proof of Proposition \ref{prop:monoidWorkBound}
    as well as the definition of $\ell$ and $r$ from the proof of Proposition \ref{thm:D1polylog}.
    We first describe a dynamic program $\prog$ for $\RangeMember{D_1}$.
   It maintains $\ell$ and $r$ for all special intervals.
    Similar to the proof of Proposition \ref{prop:monoidWorkBound}, the two crucial observations are that
    (1) a range query can be answered with the help of a constant number of special intervals, and
    (2) the change operation affects only a bounded number of special intervals per level. 

    \newcommand{\diff}{\textnormal{d}}
    Towards observation (1), recall that each (not necessarily special) interval can be split into a constant number of special intervals.
    To answer a query $\qrange(p,q)$ the program can verify whether the \emph{balance factor}, that is the sum of the differences $\ell-r$ over all special intervals splitting $[p,q]$, is zero,
    and that the balance factor is not negative for any prefix of the sequence of the special intervals.

    Towards observation (2), note that a change operation only affects the special intervals that contain the changed position and
    that $\ell$ and $r$ can be updated bottom up in the same fashion as in the proof of Proposition \ref{prop:monoidWorkBound}.

    We turn now to $\RangeMember{D_k}$.
    Here, $\prog$ maintains again $\ell$ and $r$ for all special intervals but ignoring the type of parenthesis.
    If the query asks whether the string from position $p$ to position $q$ is in $D_k$, then the program checks that inside $[p,q]$ all opening parentheses have consistent closing parentheses and
    that the ``level'' of $p$ is the same as of $q$, that is, that the number of opening parentheses is the same as the number of closing parentheses. This can be done by inspecting the  special intervals that lead from $p$ to $q$ in the canonical way.
    
    As there are a bounded number of levels in the hierarchy the work is $\bigO(n^\epsilon)$ per position and therefore $\bigO(n^{1+\epsilon})$ in total.
\end{proofsketch}

\begin{proofsketch}[of Lemma~\ref{lem:stringEqualityLinear}]
    \newcommand{\rank}{\textnormal{rank}}
    The proof is not entirely trivial, as we allow strings to have positions labelled by $\epsilon$, so we cannot simply compare the positions of the input strings one by one in parallel. \tsm{Mental note: otherwise it would be even in FO with work O(n)} To deal with positions labelled by $\epsilon$, the dynamic program maintain a bijection between the non-$\epsilon$ positions of the two given strings $s_1$ and $s_2$ such that the $i$-th non-$\epsilon$ position of $s_1$ is mapped to the $i$-th non-$\epsilon$ position in $s_2$, and vice versa.
    The strings are equal if all mapped positions have the same label, which can be checked with work $\bigO(n)$ after each change.

    Such a bijection can be defined using the functions $\rank_j(i)$ that map a \mbox{non-$\epsilon$} position $i$ in the string $s_j$ to the number $m$, if $i$ is the $m$-th non-$\epsilon$ position $i$ in~$s_j$, and the inverse $\rank_j^{-1}(i)$ of this function.
    The function $\rank_j$ for a string $s_j$ can be maintained as follows.
    If a (non-$\epsilon$) symbol is inserted at position $p$, the rank of all non-$\epsilon$ positions $i>p$ is increased by $1$ and
    the changed position $p$ gets $\rank_j(i')+1$ where $i'$ is the next non-$\epsilon$ position to the left of $p$.
    Finding this position can be done in \DynFO with $\bigO(\log n)$ work thanks to Lemma~\ref{lem:NextInKWork}.
    If a symbol is deleted from a position $p$, the rank of all non-$\epsilon$ positions $i>p$ is decreased by $1$ and
    the $\rank_j(p)$ is no longer defined.
    The inverse can be maintained in a very similar way.
    The necessary work for updating the functions $\rank_j$ and $\rank_j^{-1}$ at the changed position $p$ is bounded by $\bigO(\log n)$, for all other positions it is bounded by a constant.
    So, the overall work is $\bigO(n)$.
\end{proofsketch}

Unfortunately, if we generalize the definition of a reduced string $\red(w)$ to arbitrary many types of parentheses,
the result is in general not a string of closing parentheses followed by a string of opening ones.
The idea of \cite{FrandsenHusfeldt+1995} is now to maintain the parentheses to the left and right that remain
if we reduce the string by matching opening and closing parenthesis regardless of their type.

For a string $w$ over the alphabet $\Sigma_k = \Sigma_k^{\open} \cup \Sigma_k^{\close} = \{\open_1, \ldots, \open_k\} \cup \{\close_1, \ldots, \close_k\}$
we define a \emph{type-aware reduced string} $\reda(w)$ and a \emph{type-unaware reduced string} $\redu(w)$ as $\reda(\epsilon) = \redu(\epsilon) = \epsilon$ and
\[
    \reda(z \sigma) = 
        \begin{cases}
            z'              & \textnormal{if $\reda(z)=z'\open_i$ and $\sigma = \close_i$, for some $i \leq k$}\\
            \reda(z)\sigma   & \textnormal{else}
        \end{cases}
\]
\[
    \redu(z \sigma) = 
        \begin{cases}
            z'              & \textnormal{if $\redu(z)=z'\open_i$ and $\sigma = \close_j$, for some $i,j \leq k$}\\
            \redu(z)\sigma  & \textnormal{else.}
        \end{cases}
\]

Note that, as before, the type-unaware reduced string can always be split into a first and second part that consist only of closing and opening parentheses, respectively.
For a string $w \in \Sigma_k^\ast$, we therefore use the type-unaware reduced string to define the unmatched parentheses to the left $\leftS{w}$ and to the right~$\rightS{w}$.
We slightly abuse the notation by writing $\leftS{x}$ for a node $x$ instead of $\leftS{\str(x)}$ (as well as $\rightS{x}$, $\reda(x)$ and $\redu(x)$).

To be able to answer $\Member{D_k}$, the dynamic program we construct maintains the unmatched parentheses for every node $x$ of a tree spanning the input word, and a bit $M(x)$ that indicates whether $\reda(x) = \redu(x)$ holds.

In \cite{FrandsenHusfeldt+1995} the unmatched parentheses and the bit $M(x)$ are updated bottom up after a change operation.
For a node $x$ with children $y_1$ and $y_2$ the bit $M(x)$ is set to $1$ if both 
$M(y_1)$ and $M(y_2)$ are set to $1$ and if the unmatched opening parentheses of $\str(y_1)$ match correctly with the unmatched closing parentheses of $\str(y_2)$ to the left. The latter condition can be tested by testing whether two strings are equal, namely the string $\rightS{y_1}$ and the string that results from $\leftS{y_2}$ by reversing the string and exchanging a closing parenthesis by an opening parenthesis of the same type (assuming that both strings have equal length, otherwise surplus symbols in the longer string are ignored). 

In \cite{FrandsenHusfeldt+1995}, \StringEquality is maintained in polylogarithmic sequential time using~\cite{MehlhornSundar+1997},
in~\cite{FreydenbergerThompson2020} it is shown that maintaining whether two substrings are equal is in \DynFO with a polynomial amount of work.
Lemma~\ref{lem:stringEqualityLinear}  tells us that a linear amount of work is sufficient if one is only interested in equality of entire strings.

\begin{proofsketch}[of Lemma~\ref{lem:Dklinearithmic}]
    We describe a dynamic program $\prog$ that maintains $\Member{D_k}$, using a precomputed binary tree $T$ spanning the input word.
    For each node $x$ of $T$ the program maintains $\leftS{x}$ and $\rightS{x}$, represented by the subset of positions that constitute these strings, 
    the bit $M(x)$, as well as the string $\big(\overline{\leftS{x}}\big)^R$ that results from $\leftS{x}$ by first exchanging each closing parenthesis by an opening parenthesis of the same type, and then reversing the string.  
    Further auxiliary information is introduced later on. 
    The input word is in $D_k$ if the root of the tree satisfies $M(\textnormal{root})=1$ and $\leftS{\textnormal{root}} = \rightS{\textnormal{root}} = \epsilon$.

    We first describe how $\leftS{x}$ and $\rightS{x}$ can be maintained with work $\bigO((\log n)^3)$ for a node $x$ after a closing parenthesis is inserted at a position $p$ corresponding to a leaf below $x$. Similar to $D_1$, after such an insertion either an opening parenthesis has to be removed from $\rightS{x}$ or a closing one has to be inserted into $\leftS{x}$. In contrast to $D_1$ it is not sufficient to maintain only the number of unmatched parentheses, as we have to check whether only parenthesis of the same type are matched, so the unmatched parentheses will be maintained explicitly. But similar to $D_1$, one can infer how the unmatched parentheses need to be updated for a tree node $x$ either directly from the existing auxiliary information of $x$ and its children, or from the auxiliary information of some easy-to-identify node that lies on the path from $x$ to the leaf corresponding to the changed position $p$. 
Which case applies depends on in which ``segment'' of $\str(x)$ the change happened.

Suppose that $y_1$ and $y_2$ are the children of $x$.  We call $x$ \emph{right-heavy} if $\size{\rightS{y_1}} \le \size{\leftS{y_2}}$ and \emph{left-heavy} otherwise. Denote by $p_r(x)$ the position of the first unmatched opening parenthesis in $\rightS{x}$ (if it exists), let $p_m(x)$ be the position of the first unmatched opening parenthesis in $\rightS{y_1}$ (if it exists), and let $p_m'(x)$ be the position of the parenthesis that is matched with the parenthesis at position~$p_m$ (so, the parenthesis at position $\size{\rightS{y_1}}-1$ in $\leftS{y_2}$).

We first assume that $x$ is right-heavy. In this case we have that $p_m(x) < p'_m(x) < p_r(x)$ (if these positions exist) and consider the following four cases:
\begin{enumerate}[(1)]
    \item $p_r(x)$ is defined and $p > p_r(x)$,
    \item $p_m(x)$ is defined and $p_m(x)< p < p_m'(x)$,
    \item Cases (1) and (2) do not apply and 
			\begin{enumerate}
			 \item[(3a)] there is a node $z$ on the path from $x$ to the leaf for $p$ such that Case (2) applies for $z$, or
			 \item[(3b)] there is no such node $z$. 
			\end{enumerate}
\end{enumerate}

\nilsm{A picture might be nice}

%
    
In Case (1), that is, if $p_r(x)$ is defined and $p>p_r(x)$, the next opening parenthesis to the left of $p$ in $\rightS{x}$ has to be removed from $\rightS{x}$ since it is now matched by the new parenthesis. The problem of finding this position is an instance of \NextInK for the set of opening parentheses, which can be solved with work $\bigO(\log n)$ thanks to Lemma~\ref{lem:NextInKWork}.

In all other cases, a closing parenthesis has to be inserted into $\leftS{x}$.

In Case (2), that is, if $p_m(x)$ is defined and $p_m(x)< p < p_m'(x)$, the closing parenthesis at position $p_m'(x)$ is no longer matched by the opening parenthesis at position $p_m(x)$. It is also not matched by any other opening parenthesis, as no such unmatched parentheses exist, so it needs to be inserted into $\leftS{x}$.

In Case (3) we distinguish two subcases. For Case (3a), suppose that $z$ is the first right-heavy node on the path from $x$ to the leaf of $p$
    such that $p_m(z)< p < p_m'(z)$ holds. 
    The same closing parenthesis that has to be inserted into $\leftS{z}$ has to be inserted into $\leftS{x}$,
    that is, the parenthesis at position $p_m'(z)$. For Case (3b), that is, if there is no such node $z$, the newly inserted parenthesis itself has to be inserted into $\leftS{x}$.

We now argue the correctness of Case (3). As Cases (1) and (2) do not apply, the position $p$ lays in a substring $v$ of $\str(x)$ that either (i) begins at the first position of $\str(x)$ and ends at $p_m(x)-1$, or (ii) begins at $p_m'(x)+1$ and ends at $p_r(x)$. 
The string $v$ may be balanced, or it may have some unmatched closing parentheses, but there are no unmatched opening parentheses. 
So, $v$ can be split into the unmatched closing parentheses and the balanced substrings $v_1, \ldots, v_j$ between them, that is, $v = v_1 \sigma_1 v_2 \sigma_2 \ldots \sigma_{j-1} v_j$
    with $\sigma_i \in \Sigma_{k_i}^\close$ and $\redu(v_i) = \epsilon$, for all $i$. 

    We first consider the case that $p$ is directly before or after an unmatched parenthesis $\sigma_i$.
    In this case, 
    there is no opening parenthesis before position $p$ that the newly inserted parenthesis can match with.
    Therefore, this parenthesis has to be inserted into $\leftS{x}$.
    Observe that there is no node $z$ with $p_m(z)< p < p_m'(z)$ on the path from $x$ to the leaf of $p$, as that would imply that the inserted parenthesis could be matched. So, in this case, the reasoning above is correct.

    Now we consider the case that $p$ lies in some $v_i$.
    The balanced substring $v_i$ can be split again into minimal balanced substrings.
    Let $p_{\open}$ and $p_{\close}$ be the first and the last position of the minimal balanced substring that contains $p$.
    Because a closing parenthesis is inserted between $p_{\open}$ and $p_{\close}$,
    the parenthesis at position $p_{\open}$ gets matched by another parenthesis than before.
    Therefore, the parenthesis at position $p_{\close}$ remains unmatched and has to be inserted into $\leftS{x}$.

    To see that $p_{\close}$ is in fact $p_m'(z)$ for the first right-heavy node $z$ with $p_m(z) < p < p_m'(z)$ on the path from $x$ to the leaf of $p$, let $z'$ be the node for which $p_{\open}$ is in the left child $z'_1$ and $p_{\close}$ is in the right child $z'_2$ of $z'$.
    Since $p_{\open}$ and $p_{\close}$ matched each other before the insertion, they were part of $\rightS{z'_1}$ and $\leftS{z'_2}$, respectively.
    So, $p$ is strictly between $p_m(z')$ and $p_m'(z')$, and $z'$ must be right-heavy, since the string from $p_{\open}$ to $p_{\close}$ is well-balanced, and any prefix of $\str(z')$ that consists of positions smaller than $p_{\open}$ cannot contain unmatched opening parentheses.
    Additionally, if there was another right-heavy node $z''$ between $z'$ and $x$ for which $p$ is also between $p_m(z'')$ and $p_m'(z'')$,
    this would be a contradiction to the definition of $p_{\open}$ and $p_{\close}$ as being the first and last position of a minimal balanced substring of $v_i$ that contains $p$: the existence of such a node $z''$ would imply that there are unmatched opening parentheses before position $p_{\open}$.
    It follows that $z'=z$ and therefore $p_{\close}=p_m'(z)$ for the first right-heavy node $z$ with $p_m(z) < p < p_m'(z)$ on the path from $x$ to the leaf of $p$. 

    Now, let $x$ be a left-heavy node.
    As $p_r(x) = p_r(y_1) = p_m(x)$ (and as $p_m'(x)$ is undefined), we have one case less. The reasoning for the other cases is exactly as for right-heavy nodes. 

    The strings $\leftS{x}$ and $\rightS{x}$ need to be updated for all $\log n$ nodes $x$ on the path from the root to the leaf of $p$.
    The work for a single node is $\bigO((\log n)^2)$, the dominating factor is the case where the dynamic program has to find some specific node $z$, for which it examines all $\bigO((\log n)^2)$ pairs of nodes between $x$ and the leaf of position $p$. 
    So, overall the work to update $\leftS{x}$ and $\rightS{x}$ is $\bigO((\log n)^3)$.

    The updates after an insertion of an opening parenthesis are completely dual.
    For deletions of opening or closing parentheses, observe that we can simulate the deletion of, e.g., an opening parenthesis at position $p$
    by inserting a closing one at position $p+1$ and deleting both afterwards.
    To ensure that position $p+1$ is always free $\prog$ maintains the auxiliary structure actually on a string $w'$ of length $2 \cdot \size{w}$.
    If an opening parenthesis is inserted at position $p$ it is inserted at position $2 \cdot p - 1$ (i.e. at the $p$-th odd position) in $w'$.
    A closing parenthesis is inserted at the $p$-th even position.
    Removing both parentheses afterwards is easy.
    Because they match each other, both parentheses just have to be removed from $\leftS{x}$ and $\rightS{x}$ for every node $x$.

    Because on each change operation just one parenthesis is inserted to or removed from $\redu(x)$,
    maintaining the positions $p_r(x), p_m(x)$ and $p_m'(x)$ is easy.

    As mentioned above, $M(x)=1$ if $M(y_1)=M(y_2)=1$ and the first $t$ positions of $\rightS{y_1}$ and the last $t$ positions of $\big(\overline{\leftS{y_2}}\big)^R$ coincide, where $t$ is the length of the shorter of the two strings.
 So, $M(x)=1$ is one if for all nodes $y$ on the path from $x$ to the leaf of $p$ it holds that $M(y') = 1$ for the child $y'$ of $y$ that does not lie on this path, and that the partial string equality condition holds for the other child. Given the latter information, the necessary work to update $M$ is bounded by $\bigO((\log n)^2)$.

    As maintaining string equality needs a linear amount of work, using Lemma~\ref{lem:stringEqualityLinear}
    the overall work for maintaining the partial string equality for all the $\log n$ many nodes whose strings $\leftS{x}$ and $\rightS{x}$ are updated is $\bigO(n \cdot \log n)$.
 In sum, $\prog$ maintains $D_k$ with work $\bigO(n \cdot \log n + (\log n)^3)$.
\end{proofsketch}

\begin{proofsketch}[of Lemma~\ref{lem:DkvsStringEquality}]
    Part a) follows directly from the proof of Theorem~\ref{lem:Dklinearithmic} as the linear factor in the work bound is due to maintaining string equality for $\log n$ many nodes. 
For b), observe that two strings $s_1 = \sigma_1^1 \ldots \sigma_n^1$ and $s_2 = \sigma_1^2 \ldots \sigma_n^2$ over an alphabet of size $k$ are equal if and only if the word  $w \df \open_{\sigma_1^1} \ldots \open_{\sigma_n^1} \close_{\sigma_n^2} \ldots \close_{\sigma_1^2}$ is in~$D_k$. 
\end{proofsketch}


\end{document}
